\documentclass{sig-alternate-2013}

\pdfoutput=1

\permission{Copyright is held by the International World Wide Web Conference Committee (IW3C2). IW3C2 reserves the right to provide a hyperlink to the author's site if the Material is used in electronic media.}
\conferenceinfo{WWW 2015,}{May 18--22, 2015, Florence, Italy.}
\copyrightetc{ACM \the\acmcopyr}
\crdata{978-1-4503-3469-3/15/05. \\
http://dx.doi.org/10.1145/2736277.2741096}

\usepackage{times,amsmath,epsfig}
\usepackage[english]{babel}
\usepackage{graphicx, subfigure}
\usepackage{url}
\usepackage{epstopdf}

\newtheorem{theorem}{Theorem}

\title{Effective Techniques for Message Reduction and Load Balancing in Distributed Graph Computation}
\author{
{Da Yan{\small $~^{\#1}$}, James Cheng{\small $~^{\#2}$}, Yi Lu{\small $~^{\#3}$}, Wilfred Ng{\small $~^{*4}$} }
\vspace{1.6mm}\\
\fontsize{10}{10}\selectfont\itshape
$^{\#}$\,Department of Computer Science and Engineering, The Chinese University of Hong Kong\\
\fontsize{9}{9}\selectfont\ttfamily\upshape
\{$^{1}$\,yanda, $^{2}$\,jcheng, $^{3}$\,ylu\}@cse.cuhk.edu.hk
\vspace{1.2mm}\\
\fontsize{10}{10}\selectfont\rmfamily\itshape
$^{*}$\,Department of Computer Science and Engineering, The Hong Kong University of Science and Technology\\
\fontsize{9}{9}\selectfont\ttfamily\upshape
$^{4}$\,wilfred@cse.ust.hk
}

\begin{document}

\maketitle

\begin{abstract}
Massive graphs, such as online social networks and communication networks,  have become common today. To efficiently analyze such large graphs, many distributed graph computing systems have been developed. These systems employ the ``think like a vertex'' programming paradigm, where a program proceeds in iterations and at each iteration, vertices exchange messages with each other. However, using Pregel's simple message passing mechanism, some vertices may send/receive significantly more messages than others due to either the high degree of these vertices or the logic of the algorithm used. This forms the communication bottleneck and leads to imbalanced workload among machines in the cluster. In this paper, we propose two effective message reduction techniques: (1)vertex mirroring with message combining, and (2)an additional request-respond API. These techniques not only reduce the total number of messages exchanged through the network, but also bound the number of messages sent/received by any single vertex. We theoretically analyze the effectiveness of our techniques, and implement them on top of our open-source Pregel implementation called Pregel+. Our experiments on various large real graphs demonstrate that our  message reduction techniques significantly improve the performance of distributed graph computation.
\end{abstract}


\section{Introduction}\label{sec:intro}
With the growing interest in analyzing large real-world graphs such as online social networks, web graphs and semantic web graphs, many distributed graph computing systems~\cite{giraph,powergraph,mizan,graphlab,pregel,gps,ShangY13icde,blogel_vldb} have emerged. These systems are deployed in a shared-nothing distributed computing infrastructure usually built on top of a cluster of low-cost commodity PCs. Pioneered by Google's Pregel~\cite{pregel}, these systems adopt a vertex-centric computing paradigm, where programmers think naturally like a vertex when designing distributed graph algorithms. A Pregel-like system also takes care of fault recovery and scales to arbitrary cluster size without the need of changing the program code, both of which are indispensable properties for programs running in a cloud environment.

MapReduce~\cite{DeanG04osdi}, and its open-source implementation Hadoop, are also popularly used for large scale graph processing. However, many graph algorithms are intrinsically iterative, such as the computation of PageRank, connected components, and shortest paths. For iterative graph computation, a Pregel program is much more efficient than its MapReduce counterpart~\cite{pregel}.

\vspace{2mm}

\noindent{\bf Weaknesses of Pregel.} Although Pregel's vertex-centric computing model has been widely adopted in most of the recent distributed graph computing systems~\cite{giraph,graphlab,mizan,gps} (and also inspired the edge-centric model~\cite{powergraph}), Pregel's vertex-to-vertex message passing mechanism often causes bottlenecks in communication when processing real-world graphs.

To clarify this point, we first briefly review how Pregel performs message passing. In Pregel, a vertex $v$ can send messages to another vertex $u$ if $v$ knows $u$'s vertex ID. In most cases, $v$ only sends messages to its neighbors whose IDs are available from $v$'s adjacency list. But there also exist Pregel algorithms in which a vertex $v$ may send messages to another vertex that is not a neighbor of $v$~\cite{ppa_vldb,scc_pregel}. These algorithms usually adopt pointer jumping (or doubling), a technique that is widely used in designing PRAM algorithms~\cite{ShiloachV82jal}, to bound the number of iterations by $O(\log|V|)$, where $|V|$ refers to the number of vertices in the graph.

The problem with Pregel's message passing mechanism is that a small number of vertices, which we call {\em bottleneck vertices}, may send/receive much more messages than other vertices. A bottleneck vertex not only generates heavy communication, but also significantly increases the workload of the machine in which the vertex resides, causing highly imbalanced workload among different machines. Bottleneck vertices are common when using Pregel to process real-world graphs, mainly due to either (1)high vertex degree or (2)algorithm logic, which we elaborate more as follows.

We first consider the problem caused by high vertex degree. When a high-degree vertex sends messages to all its neighbors, it becomes a bottleneck vertex. Unfortunately, real-world graphs usually have highly skewed degree distribution, with some vertices having very high degrees. For example, in the {\em Twitter} who-follows-who graph\footnote[1]{\scriptsize http://law.di.unimi.it/webdata/twitter-2010/}, the maximum degree is over 2.99M while the average degree is only 35. Similarly, in the {\em BTC} dataset used in our experiments, the maximum degree is over 1.6M while the average degree is only 4.69.

We ran {\em Hash-Min}~\cite{cc2013icde,ppa_vldb}, a distributed algorithm for computing connected components (CCs), on the degree-skewed {\em BTC} dataset in a cluster with 1 master (Worker~0) and 120 slaves (Workers 1--120), and observed highly imbalanced workload among different workers, which we describe next. Pregel assigns each vertex to a worker by hashing the vertex ID regardless of the degree of the vertex. As a result, each worker holds approximately the same number of vertices, but the total number of neighbors in the adjacency lists (i.e., number of edges) varies greatly among different workers. In the computation of {\em Hash-Min} on {\em BTC}, we observed an uneven distribution of edge number among workers, as some workers contain more high-degree vertices than other workers. Since messages are sent along the edges, the uneven distribution of edge number also leads to an uneven distribution of the amount of communication among different workers. In Figure~\ref{msgNum_btc}, the taller blue bars indicate the total number of messages sent by each worker during the entire computation of {\em Hash-Min}, where we observe highly uneven communication workload among different workers.

Bottleneck vertices may also be generated by program logic. An example is the {\em S-V} algorithm proposed in~\cite{ppa_vldb,ShiloachV82jal} for computing CCs, which we will describe in detail in Section~\ref{ssec:sv}. In {\em S-V}, each vertex $v$ maintains a field $D[v]$ which records the vertex that $v$ is to communicate with. The field $D[v]$ may be updated at each iteration as the algorithm proceeds; and when the algorithm terminates,  vertices $v_i$ and $v_j$ are in the same CC iff $D[v_i]=D[v_j]$. Thus, during the computation, some vertex $u$ may communicate with many vertices $\{v_1, v_2, \ldots, v_k\}$ in its CC if $u=D[v_i]$, for $1 \le i \le k$. In this case, $u$ becomes a bottleneck vertex.

We ran {\em S-V} on the USA road network in a cluster with 1 master (Worker~0) and 60 slaves (Workers 1--60), and observed highly imbalanced communication workload among different workers. In Figure~\ref{msgNum_usa}, the taller blue bars indicate the total number of messages sent by each worker during the entire computation of {\em S-V}, where we can see that the communication workload is very biased (especially at Worker~0). We remark that the imbalanced communication workload is not caused by skewed vertex degree distribution, since the largest vertex degree of the USA road network is merely 9. Rather, it is because of the algorithm logic of {\em S-V}. Specifically, since the USA road network is connected, in the last round of {\em S-V}, all vertices $v$ have $D[v]$ equal to Vertex~0, indicating that they all belong to the same CC. Since Vertex~0 is hashed to Worker~0, Worker~0 sends much more messages than the other workers, as can be observed from Figure~\ref{msgNum_usa}.

In addition to the two problems mentioned above, Pregel's message passing mechanism is also not efficient for processing graphs with (relatively) high average degree due to the high overall communication cost. However, many real-world graphs such as social networks and mobile phone networks have relatively high average degree, as a person is often connected to at least dozens of people.

\vspace{1.7mm}

\noindent{\bf Our Solution.} In this paper, we solve the problems caused by Pregel's message passing mechanism with two effective message reduction techniques. The goals are to \emph{(1)mitigate the problem of imbalanced workload by eliminating bottleneck vertices}, and to \emph{(2)reduce the overall number of messages exchanged through the network}.

The first technique is called \textbf{mirroring}, which is designed to eliminate bottleneck vertices caused by high vertex degree. The main idea is to construct mirrors of each high-degree vertex in different machines, so that messages from a high-degree vertex are forwarded to its neighbors by its mirrors in local machines. Let $d(v)$ be the degree of a vertex $v$ and $M$ be the number of machines in the cluster, mirroring bounds the number of messages sent by $v$ each time to $\min\{M,d(v)\}$. If $v$ is a high-degree vertex, $d(v)$ can be up to millions, but $M$ is normally only from tens to a few hundred. We remark that ideas similar to mirroring have been adopted by existing systems~\cite{graphlab,gps}, but we find that mirroring a vertex does not always reduce the number of messages due to Pregel's use of message combiner~\cite{pregel}. Hence, we provide a theoretical analysis on which vertices should be selected for mirroring in Section~\ref{sec:mirror}.

\begin{figure}[!t]
    \centering
   \includegraphics[width=\columnwidth]{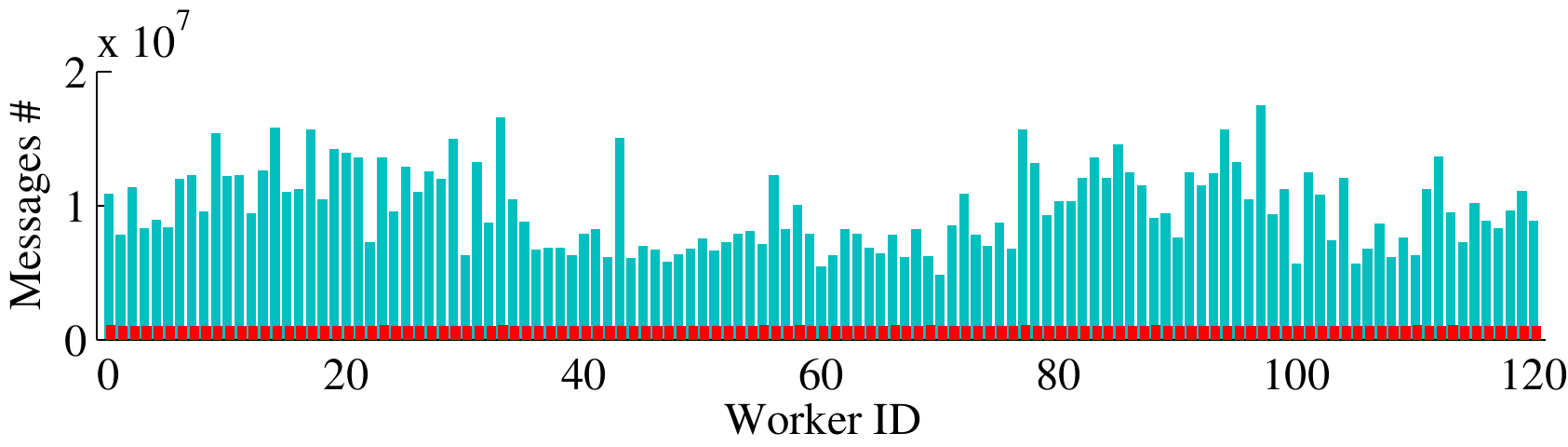}
    \vspace{-7mm}
        \caption{{\em Hash-Min} on BTC (with/without mirroring)}\label{msgNum_btc}
    \vspace{-3mm}
\end{figure}

\begin{figure}[!t]
    \centering
   \includegraphics[width=\columnwidth]{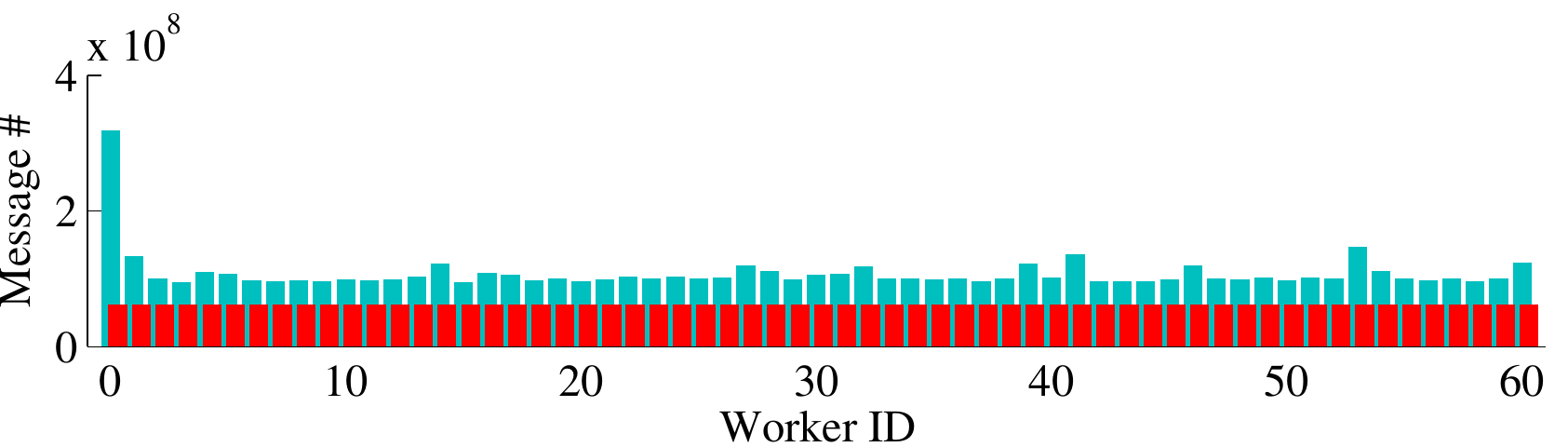}
    \vspace{-7mm}
        \caption{{\em S-V} on USA (with/without request-respond)}\label{msgNum_usa}
    \vspace{-3mm}
\end{figure}

In Figure~\ref{msgNum_btc}, the short red bars indicate the total number of messages sent by each worker when mirroring is applied to all vertices with degree at least 100. We can clearly see the big difference between the uneven blue bars (without mirroring) and the even-height short red bars (with mirroring). Furthermore, the number of messages is also significantly reduced by mirroring. We remark that the algorithm is still the same and mirroring is completely transparent to users. Mirroring reduces the running time of {\em Hash-Min} on BTC from 26.97 seconds to 9.55 seconds.

The second technique is a new \textbf{request-respond paradigm}. We extend the basic Pregel framework by an additional request-respond functionality. A vertex $u$ may request another vertex $v$ for its attribute $a(v)$, and the requested value will be available in the next iteration. The request-respond programming paradigm simplifies the coding of many Pregel algorithms, as otherwise at least three iterations are required to explicitly code each request and response process. More importantly, the request-respond paradigm effectively eliminates the bottleneck vertices resulted from algorithm logic, by bounding the number of response messages sent by any vertex to $M$. Consider the {\em S-V} algorithm mentioned earlier, where a set of $k$ vertices $\{v_1, v_2, \ldots, v_k\}$ with $D[v_i]=u$ require the value of $D[u]$ from $u$ (thus there are $k$ requests and responses). Under the request-respond paradigm, all the requests from a machine to the same target vertex are merged into one request. Therefore, at most $\min\{M, k\}$ requests are needed for the $k$ vertices and at most $\min\{M, k\}$ responses are sent from $u$. For large real-world graphs, $k$ is often orders of magnitude greater than $M$.

In Figure~\ref{msgNum_usa}, the short red bars indicate the total number of messages sent by each worker when the request-respond paradigm is applied. Again, the skewed message passing represented by the blue bars are now replaced by the even-height short red bars. In particular, Vertex~0 now only responds to the requesting workers instead of all the requesting vertices in the last round, and hence the highly imbalanced workload caused by Vertex~0 in Worker~0 is now evened out. The request-respond paradigm reduces the running time of {\em S-V} on the USA road network from 261.9 seconds to 137.7 seconds.

Finally, we remark that our experiments were run in a cluster without any resource contention, and our optimization techniques are expected to improve the overall performance of Pregel algorithms more significantly if they were run in a public data center, where the network bandwidth is lower and reducing communication overhead becomes more important.

The rest of the paper is organized as follows. We review existing parallel graph computing systems, and highlight the differences of our work from theirs,  in Section~\ref{sec:related}. In Section~\ref{sec:alg}, we describe some Pregel algorithms for problems that are common in social network analysis and web analysis. In Section~\ref{sec:chmsg}, we introduce the basic communication framework. We present the mirroring technique and the request-respond functionality in Sections~\ref{sec:mirror} and~\ref{sec:req}. Finally, we report the experimental results in Section~\ref{sec:results} and conclude the paper in Section~\ref{sec:conclude}.

\section{Background and Related Work}\label{sec:related}

\begin{figure}[!t]
    \centering
    \includegraphics[width=0.75\columnwidth]{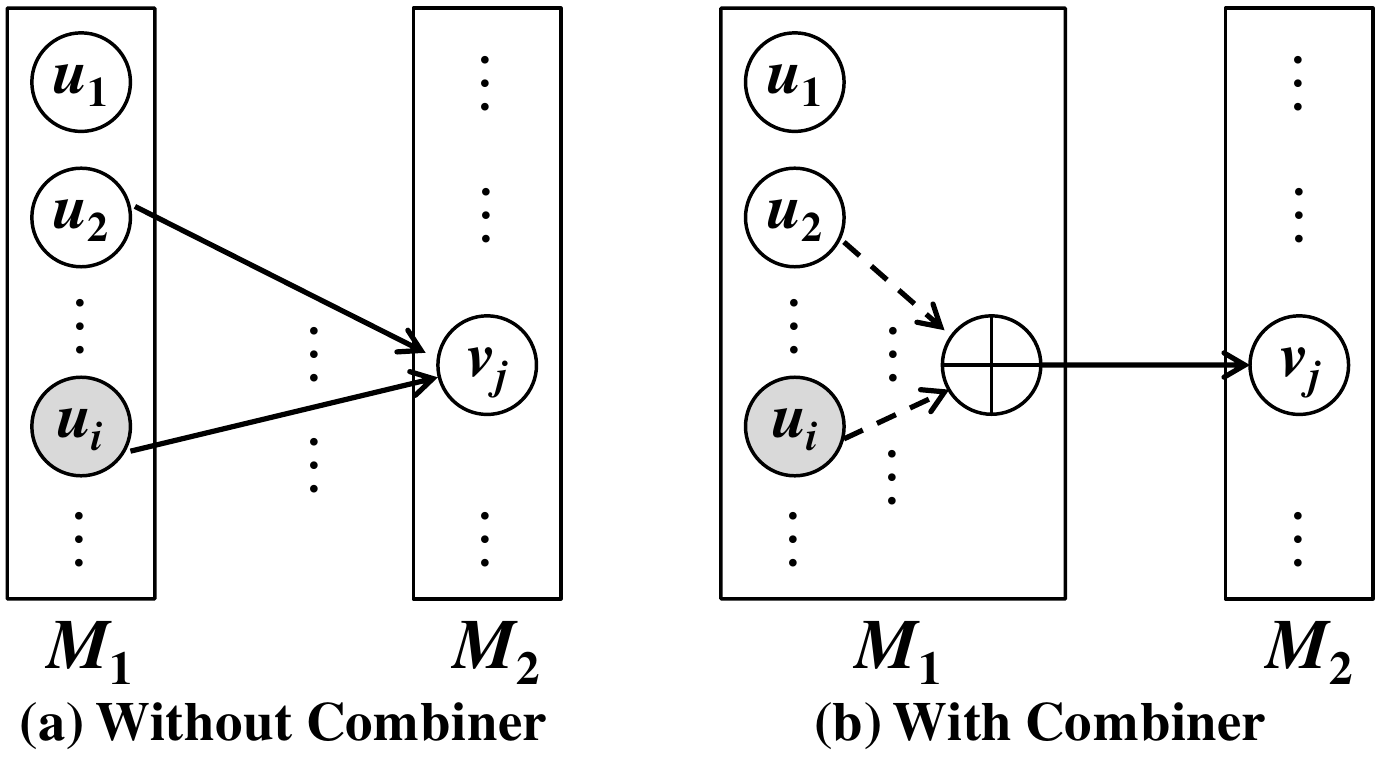}
    \vspace{-2mm}
    \caption{Illustration of combiner}\label{combiner}
    \vspace{-2mm}
\end{figure}

We first review Pregel's framework, and then discuss other related distributed graph computing systems.

\subsection{Pregel} \label{ssec:pregel}
Pregel~\cite{pregel} is designed based on the bulk synchronous parallel (BSP) model. It distributes vertices to different machines in a cluster, where each vertex $v$ is associated with its adjacency list (i.e., the set of $v$'s neighbors). A program in Pregel implements a user-defined {\em compute}() function and proceeds in iterations (called {\em supersteps}). In each superstep, the program calls {\em compute}() for each active vertex. The {\em compute}() function performs the user-specified task for a vertex $v$, such as processing $v$'s incoming messages (sent in the previous superstep), sending messages to other vertices (to be received in the next superstep), and making $v$ vote to halt. A halted vertex is reactivated if it receives a message in a subsequent superstep. The program terminates when all vertices vote to halt and there is no pending message for the next superstep.

Pregel numbers the supersteps so that a user may use the current superstep number when implementing the algorithm logic in the {\em compute}() function.  As a result, a Pregel algorithm can perform different operations in different supersteps by branching on the current superstep number.

\vspace{2mm}

\noindent{\bf Message Combiner.} Pregel allows users to implement a {\em combine}() function, which specifies how to combine messages that are sent from a machine $M_i$ to the same vertex $v$ in a machine $M_j$. These messages are combined into a single message, which is then sent from $M_i$ to $v$ in $M_j$. However, combiner is applied only when commutative and associative operations are to be applied to the messages. For example, in the PageRank computation, the messages sent to a vertex $v$ are to be summed up to compute $v$'s PageRank value; in this case, we can combine all messages sent from a machine $M_i$ to the same target vertex in a machine $M_j$ into a single message that equals their sum. Figure~\ref{combiner} illustrates the idea of combiner, where the messages sent by vertices in machine $M_1$ to the same target vertex $v_j$ in machine $M_2$ are combined into their sum before sending.

\vspace{2mm}

\noindent{\bf Aggregator.} Pregel also supports aggregator, which is useful for global communication. Each vertex can provide a value to an aggregator in {\em compute}() in a superstep. The system aggregates those values and makes the aggregated result available to all vertices in the next superstep.

\subsection{Pregel-Like Systems in JAVA}    \label{ssec:pregelLike}

Since Google's Pregel is proprietary, many open-source Pregel counterparts are developed. Most of these systems are implemented in JAVA, e.g., Giraph~\cite{giraph} and GPS~\cite{gps}. They read the graph data from Hadoop's DFS (HDFS) and write the results to HDFS. However, since object deletion is handled by JAVA's Garbage Collector (GC), if a machine maintains a huge amount of vertex/edge objects in main memory, GC needs to track a lot of objects and the overhead can severely degrade the system performance. To decrease the number of objects being maintained, JAVA-based systems maintain vertices in main memory in their binary representation. For example, Giraph organizes vertices as main memory pages, where each page is simply a byte array object that holds the binary representation of many vertices. As a result, a vertex needs to be deserialized from the page holding it before calling {\em compute}(); and after {\em compute}() completes, the updated vertex needs to be serialized back to its page. The serialization cost can be high, especially if the adjacency list is long. To avoid unnecessary serialization cost, a Pregel-like system should be implemented in a language such as C/C++, where programmers (who are system developers, not end users) manage main memory objects themselves. We implemented our Pregel+ system in C/C++.

GPS~\cite{gps} supports an optimization called large adjacency list partitioning (LALP) to handle high-degree vertices, whose idea is similar to vertex mirroring. However, GPS does not explore the performance tradeoff between vertex mirroring and message combining. Instead, it is claimed in~\cite{gps} that very small performance difference can be observed whether combiner is used or not, and thus, GPS simply does not perform sender-side message combining. Our experiments in Section~\ref{sec:results} show that sender-side message combining significantly reduces the overall running time of Pregel algorithms, and therefore, both vertex mirroring and message combining should be used to achieve better performance. As we shall see in Section~\ref{sec:mirror}, vertex mirroring and message combining are two conflicting message reduction techniques, and a theoretical analysis on their performance tradeoff is needed in order to devise a cost model for automatically choosing vertices for mirroring.

\subsection{GraphLab and PowerGraph}\label{ssec:graphlab}
GraphLab~\cite{graphlab} is another parallel graph computing system that follows a design different from Pregel. GraphLab supports asynchronous execution, and adopts a data pulling programming paradigm. Specifically, each vertex actively pulls data from its neighbors, rather than passively receives messages sent/pushed by its neighbors. This feature is somewhat similar to our request-respond paradigm, but in GraphLab, the requests can only be sent to the neighbors. As a result, GraphLab cannot support parallel graph algorithms where a vertex needs to communicate with a non-neighbor. Such algorithms are, however, quite popular in Pregel as they make use of the pointer jumping (or doubling) technique of PRAM algorithms to bound the number of iterations by $O(\log|V|)$. Examples include the {\em S-V} algorithm for computing CCs~\cite{ppa_vldb} and Pregel algorithm for computing minimum spanning forest~\cite{scc_pregel}. These algorithms can benefit significantly from our request-respond technique. Recently, several studies~\cite{sys_vldb,LuCYW15pvldb} reported that GraphLab's asynchronous execution is generally slower than its synchronous mode (that simulates Pregel's model) due to the high locking/unlocking overhead. Thus, we mainly focus on Pregel's computing model in this paper.

GraphLab also builds mirrors for vertices, which are called ghosts. However, GraphLab creates mirrors for every vertex regardless of its degree, which leads to excessive space consumption. A more recent version of GraphLab, called PowerGraph~\cite{powergraph}, partitions the graph by edges rather than by vertices. Edge partitioning mitigates the problem of imbalanced workload as the edges of a high-degree vertex are handled by multiple workers. Accordingly, a new edge-centric Gather-Apply-Scatter (GAS) computing model is used instead of the traditional vertex-centric computing model.

\section{Pregel Algorithms}  \label{sec:alg}
In this section, we describe some Pregel algorithms for problems that are common in social network analysis and web analysis, which will be used for illustrating important concepts and for performance evaluation.

We consider fundamental problems such as (1)computing connected components (or bi-connected components), which is a common preprocessing step for social network analysis~\cite{mislove2007measurement,niu2012evolution}; (2)computing minimum spanning tree (or forest), which is useful in mining social relationships~\cite{niu2012evolution}; and (3)computing PageRank, which is widely used in ranking web pages~\cite{page1999pagerank,jeh2003scaling} and spam detection\cite{gyongyi2004combating}.

For ease of presentation, we first define the graph notations used in the paper. Given an undirect graph $G=(V, E)$, we denote the neighbors of a vertex $v\in V$ by $\Gamma(v)$, and the degree of $v$ by $d(v)=|\Gamma(v)|$; if $G$ is directed, we denote the in-neighbors (out-neighbors) of a vertex $v$ by $\Gamma_{in}(v)$ ($\Gamma_{out}(v)$), and the in-degree (out-degree) of $v$ by $d_{in}(v)=|\Gamma_{in}(v)|$ ($d_{out}(v)=|\Gamma_{out}(v)|$). Each vertex $v\in V$ has a unique integer ID, denoted by $id(v)$. The diameter of $G$ is denoted by $\delta$.

\subsection{Attribute Broadcast}    \label{ssec:broadcast}

We first introduce a Pregel algorithm for {\em attribute broadcast}. Given a directed graph $G$, where each vertex $v$ is associated with an attribute $a(v)$ and an adjacency list that contains the set of $v$'s out-neighbors $\Gamma_{out}(v)$, \emph{attribute broadcast} constructs a new adjacency list for each vertex $v$ in $G$, which is defined as $\widehat{\Gamma}_{out}(v)=\{\langle u, a(u)\rangle | u \in \Gamma_{out}(v)\}$.

Put simply, {\em attribute broadcast} associates each neighbor $u$ in the adjacency list of a vertex $v$ with $u$'s attribute $a(u)$. {\em Attribute broadcast} is very useful in distributed graph computation, and it is a frequently performed key operation in many Pregel algorithms. For example, the Pregel algorithm for computing bi-connected components~\cite{ppa_vldb} requires to relabel the ID of each vertex $u$ by its preorder number in the spanning tree, denoted by $pre(u)$. {\em Attribute broadcast} is used in this case, where $a(u)$ refers to $pre(u)$.

The Pregel algorithm for {\em attribute broadcast} consists of 3 supersteps: in superstep~1, each vertex $v$ sends a message $\langle v\rangle$ to each neighbor $u \in\Gamma_{out}(v)$ to request for $a(u)$; then in superstep~2, each vertex $u$ obtains the requesters $v$ from the incoming messages, and sends the response message $\langle u, a(u)\rangle$ to each requester $v$; finally in superstep~3, each vertex $v$ collects the incoming messages to construct $\widehat{\Gamma}_{out}(v)$.

\subsection{PageRank}   \label{ssec:pagerank}
Next we present a Pregel algorithm for PageRank computation. Given a directed web graph $G=(V, E)$, where each vertex (page) $v$ links to a list of pages $\Gamma_{out}(v)$, the problem is to compute the PageRank, $pr(v)$, of each vertex $v \in V$.

Pregel's PageRank algorithm~\cite{pregel} works as follows. In superstep~1, each vertex $v$ initializes $pr(v)=1/|V|$ and distributes the value $\langle pr(v)/d_{out}(v)\rangle$ to each out-neighbor of $v$. In superstep~$i$ ($i > 1$), each vertex $v$ sums up the received values from its in-neighbors, denoted by $sum$, and computes $pr(v)=0.15/|V|+0.85\times sum$. It then distributes $\langle pr(v)/d_{out}(v)\rangle$ to each of its out-neighbors.

\subsection{Hash-Min}   \label{ssec:hashmin}
We next present a Pregel algorithm for computing connected components (CCs) in an undirected graph. We adopt the {\em Hash-Min} algorithm~\cite{cc2013icde, ppa_vldb}. Given a CC $C$, let us denote the set of vertices of $C$ by $V(C)$, and define the ID of $C$ to be $id(C)=\min\{id(v): v \in V(C)\}$. We further define the \emph{color} of a vertex $v$ as $cc(v)=id(C)$, where $v \in V(C)$. {\em Hash-Min} computes $cc(v)$ for each vertex $v \in V$, and the idea is to broadcast the smallest vertex ID seen so far by each vertex $v$, denoted by $min(v)$. When the algorithm terminates, $min(v)=cc(v)$ for each vertex $v \in V$.

We now describe the {\em Hash-Min} algorithm in Pregel framework. In superstep~1, each vertex $v$ sets $min(v)$ to be $id(v)$, broadcasts $min(v)$ to all its neighbors, and votes to halt. In each subsequent superstep, each vertex $v$ receives messages from its neighbors; let $min^*$ be the smallest ID received, if $min^* < min(v)$, $v$  sets $min(v) = min^*$ and broadcasts $min^*$ to its neighbors. All vertices vote to halt at the end of a superstep. When the process converges, all vertices have voted to halt and for each vertex $v$, we have $min(v)=cc(v)$.

\subsection{The S-V Algorithm}   \label{ssec:sv}
The {\em Hash-Min} algorithm described in Section~\ref{ssec:hashmin} requires $O(\delta)$ supersteps~\cite{ppa_vldb}, which can be slow for computing CCs in large-diameter graphs. Another Pregel algorithm proposed in~\cite{ppa_vldb} computes CCs in $O(\log |V|)$ supersteps, by adapting Shiloach-Vishkin's ({\em S-V}) algorithm for the PRAM model~\cite{ShiloachV82jal}. We use this algorithm to demonstrate how algorithm logic generates a bottleneck vertex $v$ even if $d(v)$ is small.

\begin{figure}[!t]
    \centering
    \includegraphics[width=0.77\columnwidth]{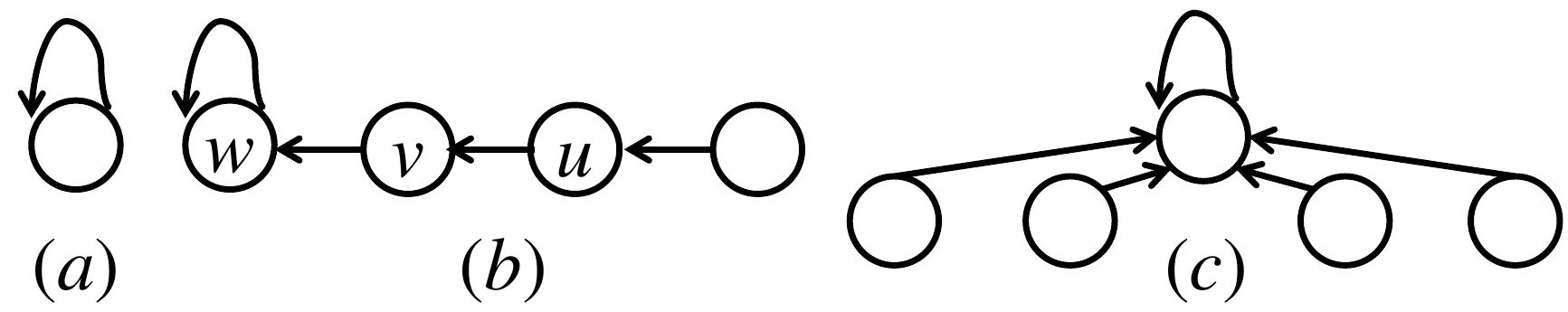}
        \vspace{-2mm}
    \caption{Forest structure of the {\em S-V} algorithm}\label{forest}
            \vspace{-3mm}
\end{figure}

In the {\em S-V} algorithm, each vertex $u$ maintains a pointer $D[u]$, which is initialized as $u$, forming a self loop as shown Figure~\ref{forest}(a). During the computation, vertices are organized into a forest such that all vertices in a tree belong to the same CC. The tree definition is relaxed a bit here to allow the tree root $w$ to have a self-loop, i.e., $D[w]=w$  (see Figures~\ref{forest}(b) and \ref{forest}(c)); while $D[v]$ of any other vertex $v$ in the tree points to $v$'s parent.

\begin{figure}[!t]
    \centering
    \includegraphics[width=\columnwidth]{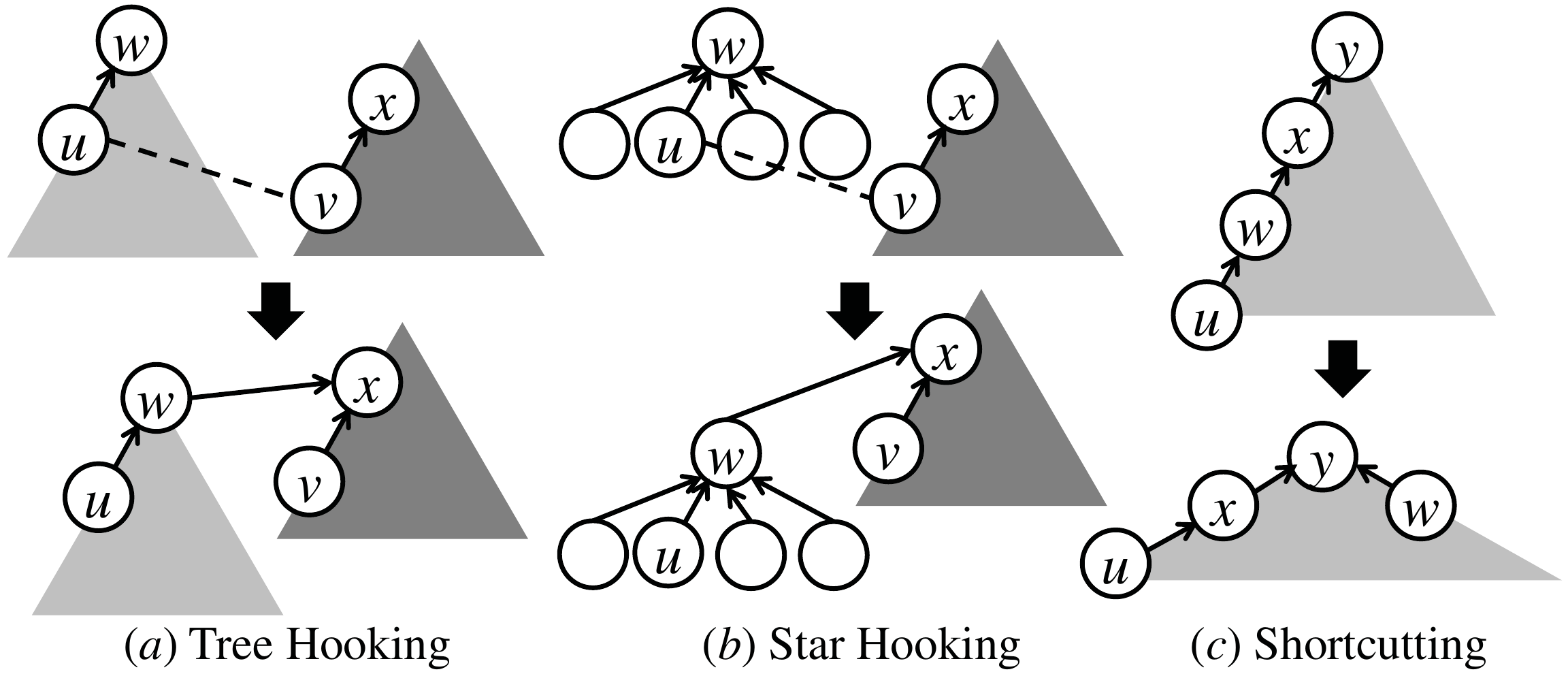}
        \vspace{-4mm}
    \caption{Key operations of the {\em S-V} algorithm}\label{svop}
\end{figure}

The {\em S-V} algorithm proceeds in rounds, and in each round, the pointers are updated in three steps (illustrated in Figure~\ref{svop}): (1){\em tree hooking}: for each edge $(u, v)$, if $u$'s parent $w=D[u]$ is a tree root, hook $w$ as a child of $v$'s parent $D[v]$, i.e., set $D[D[u]]=D[v]$; (2){\em star hooking}: for each edge $(u, v)$, if $u$ is in a star (see Figure~\ref{forest}(c) for an example of star), hook the star to $v$'s tree as in Step~(1), i.e., set $D[D[u]]=D[v]$; (3){\em shortcutting}: for each vertex $v$, move vertex $v$ and its descendants closer to the tree root, by hooking $v$ to the parent of $v$'s parent, i.e., setting $D[v]=D[D[v]]$. The above three steps execute in rounds, and the algorithm ends when every vertex is in a star.

Due to the shortcutting operation, the {\em S-V} algorithm creates flattened trees (e.g., stars) with large fan-out towards the end of the execution. As a result, a vertex $w$ may have many children $u$ (i.e., $D[u]=w$), and each of these children $u$ requests $w$ for the value of $D[w]$. This renders $w$ a bottleneck vertex. In particular, in the last round of the {\em S-V} algorithm, all vertices $v$ in a CC $C$ have $D[v]=id(C)$, and they all send requests to the vertex $w=id(C)$ for $D[w]$. In the basic Pregel framework, $w$ receives $|V(C)|$ requests and sends $|V(C)|$ responses, which leads to skewed workload when $|V(C)|$ is large.

\subsection{Minimum Spanning Forest}\label{ssec:msf}
The Pregel algorithm proposed by~\cite{scc_pregel} for minimum spanning forest (MSF) computation is another example that shows how algorithm logic can generate bottleneck vertices. This algorithm proceeds in iterations, where each iteration consists of three steps, which we describe below.

\begin{figure}[!t]
    \centering
    \includegraphics[width=0.6\columnwidth]{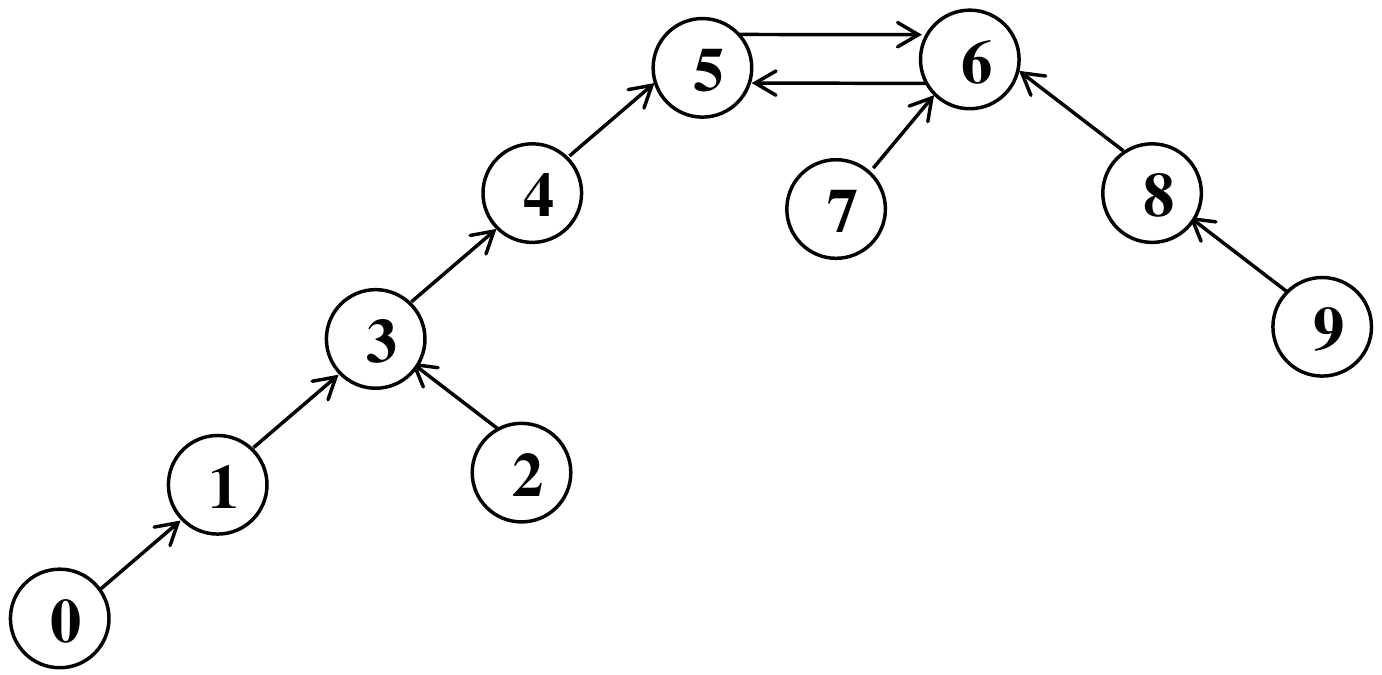}
        \vspace{-3mm}
    \caption{Conjoined Tree}\label{conjoin}
            \vspace{-3mm}
\end{figure}

In Step~(1), each vertex $v$ picks an edge with the minimum weight. The vertices and their picked edges form disjoint subgraphs, each of which is a conjoined-tree: two trees with their roots joined by a cycle. Figure~\ref{conjoin} illustrates the concept of a conjoined-tree, where the edges are those picked in Step~(1). The vertex with the smaller ID in the cycle of a conjoined-tree is called the supervertex of the tree (e.g., vertex~5 is the supervertex in Figure~\ref{conjoin}), and the other vertices are called the subvertices.


In Step~(2), each vertex finds the supervertex of the conjoined-tree it belongs to, which is accomplished by pointer jumping. Specifically, each vertex $v$ maintains a pointer $D[v]$; suppose that $v$ picks edge $(v, u)$ in Step~(1), then the value of $D[v]$ is initialized as $u$. Each vertex $v$ then sends request to $w=D[v]$ for $D[w]$. Initially, the actual supervertex $s$ (e.g., vertex~5 in Figure~\ref{conjoin}) and its neighbor $s'$ in the cycle (e.g. vertex~6 in Figure~\ref{conjoin}) see that they have sent each other messages and detect that they are in the cycle. Vertex $s$ then sets itself as the supervertex (i.e., sets $D[s]=s$) due to $s<s'$, before responding $D[s]=s$ to the requesters (while $D[s']=s$ remains for $s'$ since $s'>s$). For any other vertex $v$, it receives response $D[w]$ from $w=D[v]$ and updates $D[v]$ to be $D[w]$. This process is repeated until convergence, upon when $D[v]$ records the supervertex $s$ for all vertices $v$.

In Step~(3), each vertex $v$ sends request to each neighbor $u\in\Gamma(v)$ for its supervertex $D[u]$, and removes edge $(v, u)$ if $D[v]=D[u]$ (i.e., $v$ and $u$ are in the same conjoined-tree); $v$ then sends the remaining edges (to vertices in other conjoined-trees) to the supervertex $D[v]$. After this step, all subvertices are condensed into their supervertex, which constructs an adjacency list of edges to the other supervertices from those edges sent by its subvertices.

We consider an improved version of the above algorithm that applies the Storing-Edges-At-Subvertices (SEAS) optimization of~\cite{scc_pregel}. Specifically, instead of having the supervertex merge and store all cross-tree edges, the SEAS optimization stores the edges of a supervertex in a distributed fashion among all of its subvertices. As a result, if a supervertex $s$ is merged into another supervertex, it has to notify its subvertices of the new supervertex they belong to. This is accomplished by having each vertex $v$ send request to its supervertex $D[v]=s$ for $D[s]$. Since smaller conjoined-trees are merged into larger ones, a supervertex $s$ may have many subvertices $v$ towards the end of the execution, and they all request for $D[s]$ from $s$, rendering $s$ a bottleneck vertex.

\section{Basic Communication Framework}   \label{sec:chmsg}
When considering on which system we should implement our message reduction techniques, we decided to implement a new \emph{open-source Pregel system in C/C++}, called \textbf{Pregel+}, to avoid the pitfalls of a JAVA-based system described in Section~\ref{ssec:pregelLike}. Other reasons for a new Pregel implementation include: (1)GPS does not perform sender-side message combining, while our work studies effective message reduction techniques in a system that adheres to Pregel's framework, where message combining is supported; (2)Giraph has been shown to have inferior performance in recent performance evaluation of graph-parallel systems~\cite{sys_sigmod1,sys_bigdata,guo2013well,sys_vldb,sys_sigmod2} and also in our experiments; (3)other existing graph computing systems are also not suitable as described in Sections~\ref{ssec:graphlab} and~\ref{ssec:othersystems}.

We first introduce the basic communication framework of Pregel+. Our two new message reduction techniques to be introduced in Sections~\ref{sec:mirror} and~\ref{sec:req} further extend the basic communication framework.

We use the term ``worker'' to represent a computing unit, which can be a machine or a thread/process in a machine. For ease of discussion, we assume that each machine runs only one worker but the concepts can be straightforwardly generalized.

In Pregel+, each worker is simply an MPI (Message Passing Interface) process and communications among different processes are implemented using MPI's communication primitives. Each worker maintains a \emph{message channel}, $Ch_{msg}$, for exchanging the vertex-to-vertex messages. In the {\em compute}() function, if a vertex sends a message $msg$ to a target vertex $v_{tgt}$, the message is simply added to $Ch_{msg}$. Like in Google's Pregel, messages in $Ch_{msg}$ are sent to the target workers in batches before the next superstep begins. Note that if a message $msg$ is sent from worker $M_i$ to vertex $v_{tgt}$ in worker $M_j$, the ID of the target $v_{tgt}$ should be sent along with $msg$, so that when $M_j$ receives $msg$, it knows which vertex $msg$ should be directed to.

The operation of the message channel $Ch_{msg}$ is directly related to the communication cost and hence affects the overall performance of the system. We tested different ways of implementing $Ch_{msg}$, and the most efficient one is presented in Figure~\ref{worker}. We assume that a worker maintains $N$ vertices, $\{v_1, v_2,$ $\ldots, v_N\}$. The message channel $Ch_{msg}$ associates each vertex $v_i$ with an incoming message buffer $I_i$. When an incoming message $msg_1$ directed to vertex $v_i$ arrives, $Ch_{msg}$ looks up a hash table $T_{in}$ for the incoming message buffer $I_i$ using $v_i$'s ID. It then appends $msg_1$ to the end of $I_i$. The lookup table $T_{in}$ is static unless graph mutation occurs, in which case updates to $T_{in}$ may be required. Once all incoming messages are processed, {\em compute}() is called for each active vertex $v_i$ with the messages in $I_i$ as the input.

\begin{figure}[!t]
    \centering
    \includegraphics[width=0.9\columnwidth]{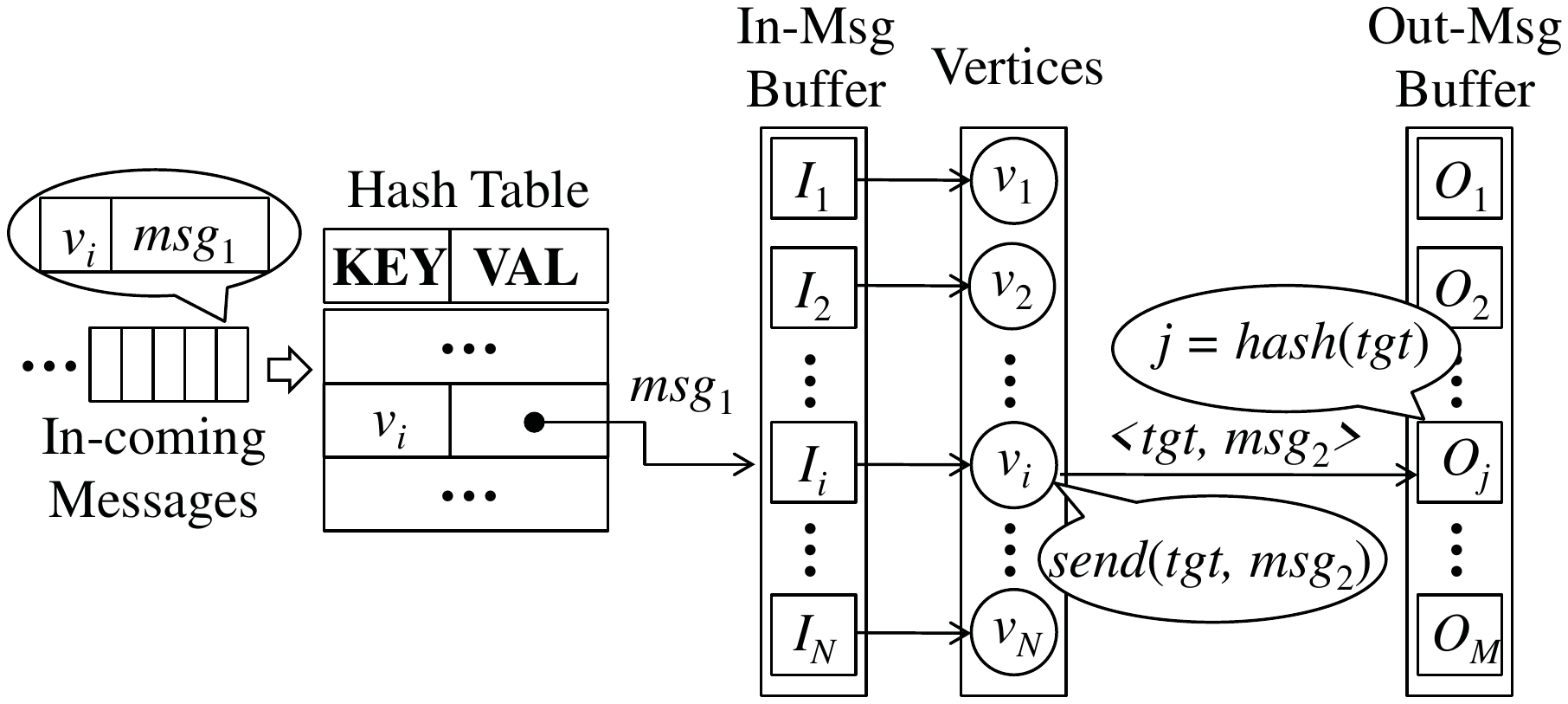}
    \vspace{-1mm}
        \caption{Illustration of Message Channel, $Ch_{msg}$}\label{worker}
    \vspace{-1mm}
\end{figure}

A worker also maintains $M$ \emph{outgoing message buffers} (where $M$ is the number of workers), one for each worker $M_j$ in the cluster, denoted by $O_j$. In {\em compute}(), a vertex $v_i$ may send a message $msg_2$ to another vertex with ID $tgt$. Let $hash(.)$ be the hash function that computes the worker ID of a vertex from its vertex ID, then the target vertex is in worker $M_{hash(tgt)}$. Thus, $msg_2$ (along with $tgt$) is appended to the end of the buffer $O_{hash(tgt)}$. Messages in each buffer $O_j$ are sent to worker $M_j$ in batch. If a combiner is used, the messages in a buffer $O_j$ are first grouped (sorted) by target vertex IDs, and messages in each group are combined into one message using the combiner logic before sending.

\section{The Mirroring Technique}   \label{sec:mirror}
The mirroring technique is designed to eliminate bottleneck vertices caused by high vertex degree.

Given a high-degree vertex $v$, we construct a mirror for $v$ in any worker in which some of $v$'s neighbors reside. When $v$ needs to send a message, e.g., the value of its attribute, $a(v)$, to its neighbors, $v$ sends $a(v)$ to its mirrors. Then, each mirror forwards $a(v)$ to the neighbors of $v$ that reside in the same local worker as the mirror, without any message passing.

Figure~\ref{mirror} illustrates the idea of mirroring. Assume that $u_i$ is a high-degree vertex residing in worker machine $M_1$, and $u_i$ has neighbors $\{v_1, v_2, \dots, v_j\}$ residing in machine $M_2$ and neighbors $\{w_1, w_2, \dots, w_k\}$ residing in machine $M_3$. Suppose that $u_i$ needs to send a message $a(u_i)$ to the $j$ neighbors in $M_2$ and $k$ neighbors in $M_3$. Figure~\ref{mirror}(a) shows how $u_i$ sends $a(u_i)$ to its neighbors in $M_2$ and $M_3$ using Pregel's vertex-to-vertex message passing. In total, $(j+k)$ messages are sent, one for each neighbor. To apply mirroring, we construct a mirror for $u_i$ in $M_2$ and $M_3$, as shown by the two squares (with label $u_i$) in Figure~\ref{mirror}(b). In this way, as illustrated in Figure~\ref{mirror}(b),  $u_i$ only needs to send $a(u_i)$ to the two mirrors in $M_2$ and $M_3$. Then, each mirror forwards $a(u_i)$ to $u_i$'s neighbors locally in $M_2$ and $M_3$ without any network communication. In total, only two messages are sent through the network, which not only tremendously reduces the communication cost, but also eliminates the imbalanced communication load caused by $u_i$.

\begin{figure}[!t]
    \centering
    \includegraphics[width=0.9\columnwidth]{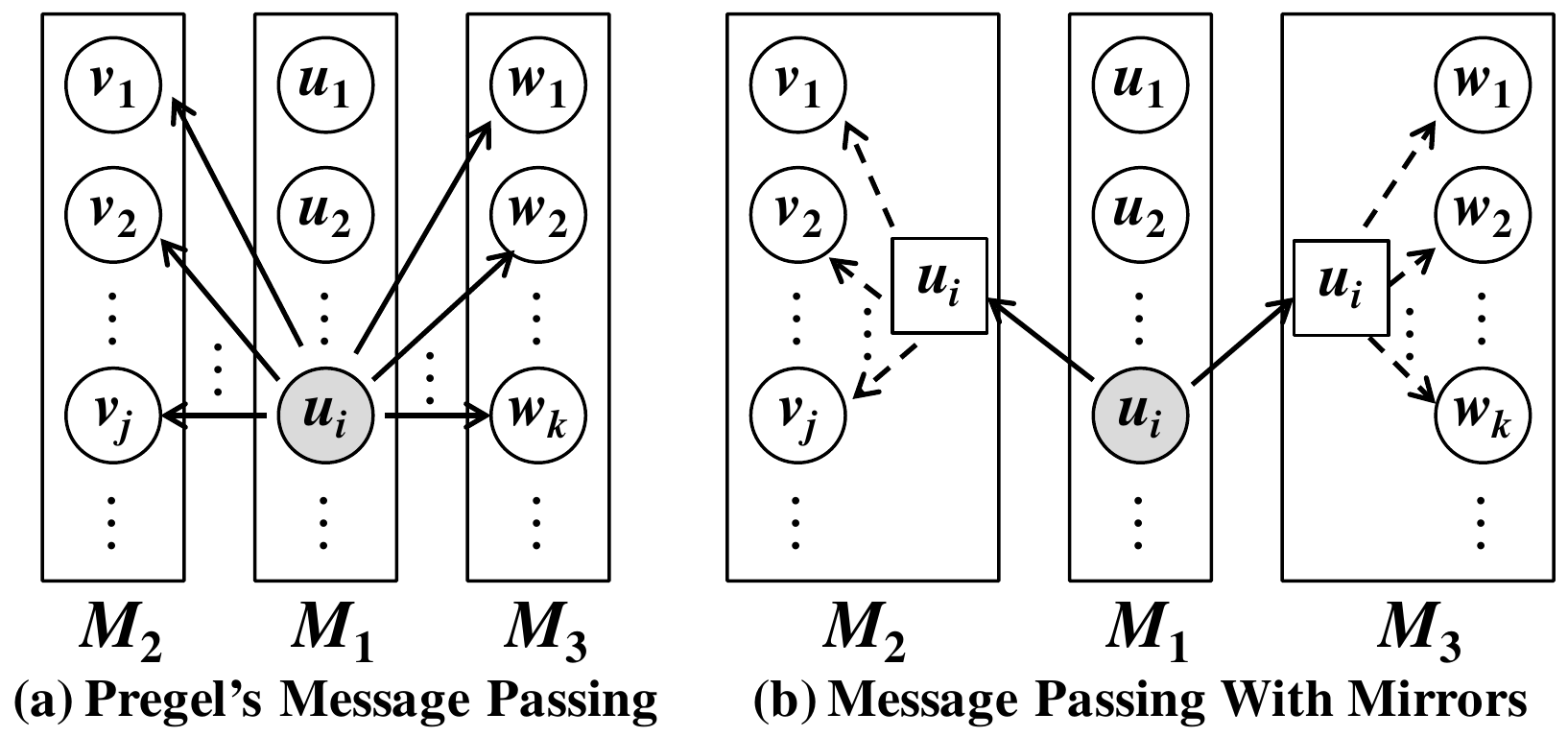}
\vspace{-2mm}
    \caption{Illustration of Mirroring}\label{mirror}
\end{figure}

We formalize the effectiveness of mirroring for message reduction by the following theorem.


\begin{theorem}\label{th:mirrorBound}
Let $d(v)$ be the degree of a vertex $v$ and $M$ be the number of machines. Suppose that $v$ is to deliver a message $a(v)$ to all its neighbors in one superstep. If mirroring is applied on $v$, then the total number of messages sent by $v$ in order to deliver $a(v)$ to all its neighbors is bounded by $\min\{M,d(v)\}$.
\end{theorem}


\begin{proof}
The proof follows directly from the fact that $v$ only needs to send one message $a(v)$ to each of its mirrors in other machines and there are at most $\min\{M,d(v)\}$ mirrors of $v$.
\end{proof}


\noindent{\bf Mirroring Threshold.} The mirroring technique is transparent to programmers. But we can allow users to specify a mirroring threshold $\tau$ such that mirroring is applied to a vertex $v$ only if $d(v) \ge \tau$ (we will see shortly that $\tau$ can be automatically set by a cost model following the result of Theorem~\ref{th:mirror}).  If a vertex has degree less than $\tau$, it sends messages through the normal message channel $Ch_{msg}$ as usual. Otherwise, the vertex only sends messages to its mirrors, and we call this message channel as the \emph{mirroring message channel}, or $Ch_{mir}$ in short. In a nutshell, a message is sent either through $Ch_{msg}$ or $Ch_{mir}$, depending on the degree of the sending vertex.

Figure~\ref{mirror_combiner} illustrates the concepts of $Ch_{msg}$ and $Ch_{mir}$, where we only consider the message passing between two machines $M_1$ and $M_2$. The adjacency lists of vertices $u_1$, $u_2$, $u_3$ and $u_4$ in $M_1$ are shown in Figure~\ref{mirror_combiner}(a), and we consider how they send messages to their common neighbor $v_2$ residing in machine $M_2$. Assume that $\tau=3$, then as Figure~\ref{mirror_combiner}(b) shows, $u_1$, $u_2$ and $u_3$ send their messages, $a(u_1)$, $a(u_2)$ and $a(u_3)$, through $Ch_{msg}$, while $u_4$ sends its message $a(u_4)$ through $Ch_{mir}$.

\begin{figure}[!t]
    \centering
    \includegraphics[width=0.83\columnwidth]{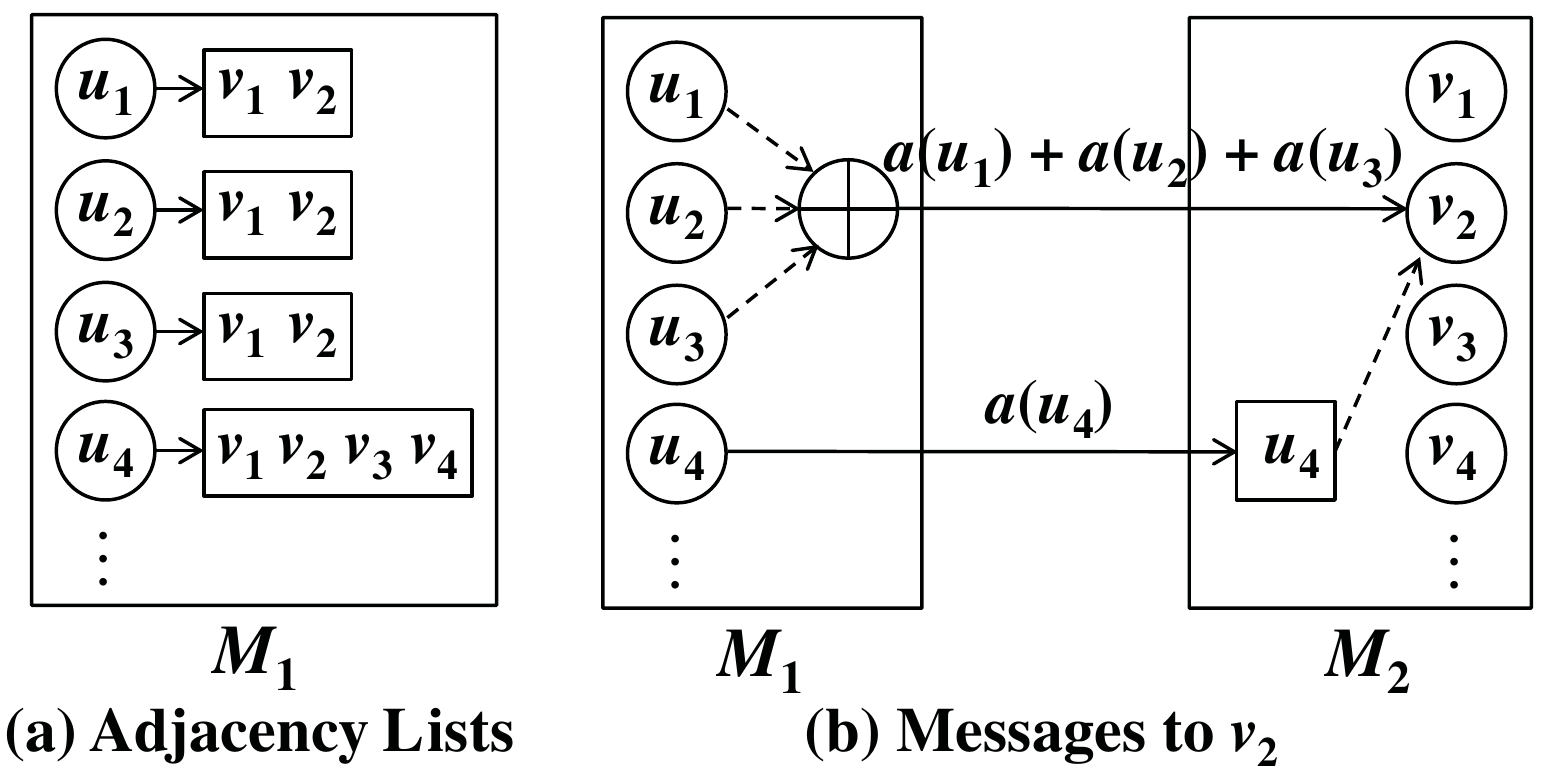}
    \vspace{-2mm}
    \caption{Mirroring v.s.\ Message Combining}\label{mirror_combiner}
    \vspace{-1mm}
\end{figure}

\vspace{2mm}

\noindent{\bf Mirroring v.s.\ Message Combining.} Now let us assume that the messages are to be applied with commutative and associative operations at the receivers' side, e.g., the message values are to be summed up as in PageRank computation. In this case, a combiner can be applied on the message channel $Ch_{msg}$. However, the \emph{receiver-centric} message combining is not applicable to the \emph{sender-centric} channel $Ch_{mir}$. For example, in Figure~\ref{mirror_combiner}(b), when $u_4$ in $M_1$ sends $a(u_4)$ to its mirror in $M_2$, $u_4$ does not need to know the receivers (i.e., $v_1$, $v_2$, $v_3$ and $v_4$); thus, its message to $v_2$ cannot be combined with those messages from $u_1$, $u_2$ and $u_3$ that are also to be sent to $v_2$. In fact, $u_4$ only holds a list of the machines that contain $u_4$'s neighbors, i.e. $\{M_2\}$ in this example, and $u_4$'s neighbors $v_1$, $v_2$, $v_3$ and $v_4$ that are local to $M_2$ are connected by $u_4$'s mirror in $M_2$.

It may appear that $u_4$'s message to its mirror is wasted, because if we combine $u_4$'s message with those messages from $u_1$, $u_2$ and $u_3$, then we do not need to send it through $Ch_{mir}$. However, we note that a high-degree vertex like $u_4$ often has many neighbors in another worker machine, e.g., $v_1$, $v_3$ and $v_4$ in addition to $v_2$ in this example, and the message is not wasted since the message is also forwarded to $v_3$ and $v_4$, which are not the neighbors of any other vertex in $M_1$. 

\vspace{2mm}

\noindent{\bf Choice of Mirroring Threshold.} The above discussion shows that there are cases where mirroring is useful, but it does not give any formal guideline as to when exactly mirroring should be applied. To this end, we conduct a theoretical analysis below on the interplay between mirroring and message combining. Our result shows that mirroring is effective even when message combiner is used.


\begin{theorem}\label{th:mirror}
Given a graph $G=(V, E)$ with $n=|V|$ vertices and $m=|E|$ edges, we assume that the vertex set is evenly partitioned among $M$ machines (e.g., by hashing as in Pregel) and each machine holds $n/M$ vertices. We further assume that the neighbors of a vertex in $G$ are randomly chosen among $V$, and the average degree $deg_{avg}=m/n$ is a constant. Then, mirroring should be applied to a vertex $v$ if $v$'s degree is at least $(M\cdot\exp\{deg_{avg}/M\})$.
\end{theorem}


\begin{proof}
Consider a machine $M_i$ that contains a set of $n/M$ vertices, $V_i=\{v_1, v_2, \ldots, v_{n/M}\}$, where each vertex $v_j$ has $\ell_j$ neighbors for $1 \le j \le n/M$. Let us focus on a specific vertex $v_j$ in $M_i$, and infer how large $\ell_j$ should be so that applying mirroring on $v_j$ can reduce the overall communication even when a combiner is used.

Consider an application where all vertices send messages to all their neighbors in each superstep, such as in PageRank computation. Further consider vertex  $u\in\Gamma_{out}(v_j)$. If another vertex $v_k \in V_i \setminus \{v_j\}$ sends messages through $Ch_{msg}$ and $v_k$ also has $u$ as its neighbor, then $v_j$'s message to $u$ is wasted since it can be combined with $v_k$'s message to $u$. We assume the worst case where all vertices in $V_i \setminus \{v_j\}$ send messages through $Ch_{msg}$. Since the neighbors of a vertex in $G$ are randomly chosen among $V$, we have
\begin{displaymath}
\Pr\{u\in\Gamma_{out}(v_k)\}=\ell_k/n,
\end{displaymath}
and therefore,
\begin{eqnarray*}
& & \Pr\{\mbox{$v_j$'s message to $u$ is not wasted}\}\\
& = & \prod_{v_k\in V_i \setminus \{v_j\}} \Pr\{u\not\in\Gamma_{out}(v_k)\}\ =\ \prod_{v_k\in V_i \setminus \{v_j\}}\left(1-\frac{\ell_k}{n}\right).
\end{eqnarray*}
We regard each $\ell_k$ as a random variable whose value is chosen independently from a degree distribution (e.g., power-law degree distribution) with expectation $E[\ell_k]=m/n=deg_{avg}$. Then, the expectation of the above equation is given by
\begin{eqnarray*}
& & E\left[\prod_{v_k\in V_i \setminus \{v_j\}}\left(1-\frac{\ell_k}{n}\right)\right]\ =\ \prod_{v_k\in V_i \setminus \{v_j\}}E\left[1-\frac{\ell_k}{n}\right]\\
& = & \prod_{v_k\in V_i \setminus \{v_j\}}\left(1-\frac{E[\ell_k]}{n}\right)\ =\ \prod_{v_k\in V_i \setminus \{v_j\}}\left(1-\frac{deg_{avg}}{n}\right)\\
& \ge & \prod_{v_k\in V_i}\left(1-\frac{deg_{avg}}{n}\right)\ =\ \left(1-\frac{deg_{avg}}{n}\right)^{n/M}.
\end{eqnarray*}

For large graphs, we have
\begin{eqnarray*}
& & \Pr\{\mbox{$v_j$'s message to $u$ is not wasted}\}\\
& \approx & \lim_{n\to\infty}\left(1-\frac{deg_{avg}}{n}\right)^{n/M}\ =\ \exp\left\{-\frac{deg_{avg}}{M}\right\},
\end{eqnarray*}
where the last step is derived from the fact that $\lim_{n\to\infty}(1-1/n)^n=e^{-1}$.

According to the above discussion, the expected number of $v_j$'s neighbors that are not the neighbors of any other vertex(es) in $M_i$ is equal to $\ell_j\cdot\exp\{-deg_{avg}/M\}$. In other words, if mirroring is not used, $v_j$ needs to send at least $\ell_j\cdot\exp\{-deg_{avg}/M\}$ messages that are not wasted. On the other hand, if mirroring is used, $v_j$ sends at most $M$ messages, one to each mirror. Therefore, mirroring reduces the number of messages if $\ell_j\cdot\exp\{-deg_{avg}/M\}\geq M$, or equivalently, $\ell_j\geq M\cdot\exp\{deg_{avg}/M\}$. To conclude, choosing $\tau=M\cdot\exp\{deg_{avg}/M\}$ as the degree threshold reduces the communication cost.
\end{proof}


Theorem~\ref{th:mirror} states that the choice of $\tau$ depends on the number of workers, $M$, and the average vertex degree, $deg_{avg}$. A cluster usually involves tens to hundreds of workers, while the average degree $deg_{avg}$ of a large real world graph is mostly below 50. Consider the scenario where $M=100$ and $deg_{avg}\leq 50$, then $\tau\leq100e^{0.5}$$=$$165$. This shows that mirroring is effective even for vertices whose degree is not very high. We remark that Theorem~\ref{th:mirror} makes some simplified assumption (e.g., $G$ being a random graph) for ease of analysis, which may not be accurate for a real graph. However, our experiments in Section~\ref{result:mirror} show that Theorem~\ref{th:mirror} is effective on real graphs.

\vspace{2mm}

\noindent{\bf Mirror Construction.} Pregel+ constructs mirrors for all vertices $v$ with $\Gamma_{out}(v)\geq\tau$ after the input graph is loaded and before the iterative computation, although mirror construction can also be pre-computed offline like GraphLab's ghost construction. Specifically, the neighbors in $v$'s adjacency list $\Gamma_{out}$ is grouped by the workers in which they reside. Each group is defined as $N_i=\{u\in\Gamma_{out}(v)\ |\ hash(u)=M_i\}$. Then, for each group $N_i$, $v$ sends $\langle v; N_i\rangle$ to worker $M_i$, and $M_i$ constructs a mirror of $v$ with the adjacency list $N_i$ locally in $M_i$. Each vertex $v_j\in N_i$ also stores the address of $v_j$'s incoming message buffer $I_j$ so that messages can be directly forwarded to $v_j$ by $v$'s mirror in $M_i$.

During graph computation, a vertex $v$ sends message $\langle v, a(v)\rangle$ to its mirror in worker $M_i$. On receiving the message, $M_i$ looks up $v$'s mirror from a hash table using $v$'s ID (similar to $T_{in}$ described in Section~\ref{sec:chmsg}). The message value $a(v)$ is then forwarded to the incoming message buffers of $v$'s neighbors locally in $M_i$.

\vspace{2mm}

\noindent{\bf Handling Edge Fields.} There are some minor changes to Pregel's programming interface for applying mirroring. In Pregel's interface, a vertex calls {\em send\_msg}$(tgt, msg)$ to send an arbitrary message $msg$ to a target vertex $tgt$. With mirroring, a vertex $v$ sends a message containing the value of its attribute $a(v)$ to all its neighbors by calling {\em broadcast}$(a(v))$ instead of calling {\em send\_msg}$(u, a(v))$ for each neighbor $u\in\Gamma_{out}(v)$.

Consider the algorithms described in Section~\ref{sec:alg}. For PageRank, a vertex $v$ simply calls {\em broadcast}$(pr(v)/|\Gamma_{out}(v)|)$; while for {\em Hash-Min}, $v$ calls {\em broadcast}$(min(v))$.

However, there are applications where the message value is not only decided by the sender vertex $v$'s state, but also by the edge that the message is sent along. For example, in Pregel's algorithm for single-source shortest path (SSSP) computation~\cite{pregel}, a vertex sends $(d(v)+\ell(v, u))$ to each neighbor $u\in\Gamma_{out}(v)$, where $d(v)$ is an attribute of $v$ estimating the distance from the source, and $\ell(v, u)$ is an attribute of its out-edge $(v, u)$ indicating the edge length.

To support applications like SSSP, Pregel+ requires that each edge object supports a function {\em relay}$(msg)$, which specifies how to update the value of $msg$ before $msg$ is added to the incoming message buffer $I_i$ of the target vertex $v_i$. If $msg$ is sent through $Ch_{msg}$, {\em relay}$(msg)$ is called on the sender-side before sending. If $msg$ is sent through $Ch_{mir}$, {\em relay}$(msg)$ is called on the receiver-side when the mirror forwards $msg$ to each local neighbor (as the edge field is maintained by the mirror). For example, in Figure~\ref{mirror_combiner}, {\em relay}$(msg)$ is called when $msg$ is passed along a dashed arrow.

By default, {\em relay}$(msg)$ does not change the value of $msg$. To support SSSP, a vertex $v$ calls {\em broadcast}$(d(v))$ in {\em compute}(), and meanwhile, the function {\em relay}$(msg)$ is overloaded to add the edge length $\ell(v, u)$ to $msg$, which updates the value of $msg$ to the required value $(d(v)+\ell(v, u))$.

\vspace{2mm}

\noindent{\bf Summary of Contributions.} GPS does not use message combining, and therefore, its LALP technique are not as effective as our mirroring technique that is reinforced with message combiner. GraphLab's ghost vertex technique creates mirrors for all vertices regardless of the vertex degree, and thus it is also not as effective as our mirroring technique. As far as we know, this is the first work that considers the integration of vertex mirroring and message combining in Pregel's computing model. In addition, we also identified the tradeoff between vertex mirroring and message combining in message reduction, and provided a cost model to automatically select vertices for mirroring so as to minimize the number of messages. As we shall see in our experiments in Section~\ref{result:mirror}, the mirroring threshold computed by our cost model in Theorem~\ref{th:mirror} achieves near-optimal performance. In addition, we also cope with the case where the message value depends on the edge field, which is not supported by GPS's LALP technique.

\section{The Request-Respond Paradigm} \label{sec:req}

In Sections~\ref{sec:intro}, \ref{ssec:sv} and~\ref{ssec:msf}, we have shown that bottleneck vertices can be generated by algorithm logic even if the input graph has no high-degree vertices. For handling such bottleneck vertices, the mirroring technique of Section~\ref{sec:mirror} is not effective. To this end, we design our second message reduction technique, which extends the basic Pregel framework with a new \emph{request-respond} functionality.

We illustrate the concept using the algorithms described in Section~\ref{sec:alg}. Using the request-respond API, {\em attribute broadcast} in Section~\ref{ssec:broadcast} is straightforward to implement: in superstep~1, each vertex $v$ sends requests to each neighbor $u \in\Gamma_{out}(v)$ for $a(u)$; in superstep~2, the vertex $v$ simply obtains $a(u)$ responded by each neighbor $u$, and constructs $\widehat{\Gamma}_{out}(v)$. Similarly, for the {\em S-V} algorithm in Section~\ref{ssec:sv}, when a vertex $v$ needs to obtain $D[w]$ from vertex $w=D[v]$, it simply sends a request to $w$ so that $D[w]$ can be used in the next superstep; for the MSF algorithm in Section~\ref{ssec:msf}, a vertex $v$ simply sends a request to its supervertex $D[v]=s$ so that $D[s]$ can be used to update $D[v]$ in the next superstep.

\vspace{2mm}

\noindent{\bf Request-Respond Message Channel.} \  We now explain in detail how Pregel+ supports the request-respond API. The request-respond paradigm supports all the functionality of Pregel. In addition, it supplements the vertex-to-vertex message channel $Ch_{msg}$ with a \emph{request-respond message channel}, denoted by $Ch_{req}$.

Figure~\ref{reqresp} illustrates how requests and responses are exchanged between two machines $M_i$ and $M_j$ through $Ch_{req}$. Specifically, each machine maintains $M$ request sets, where $M$ is the number of machines, and each request set $S_{\mbox{\scriptsize to }k}$ stores the requests to vertices in machine $M_k$. In a superstep, a vertex $v$ in machine $M_j$ may call {\em request}($u$) in its {\em compute}() function to send request to vertex $u$ for its attribute value $a(u)$ (which will be used in the next superstep). Let $hash(u)=i$, then the requested vertex $u$ is in machine $M_i$, and hence $u$ is added to the request set $S_{\mbox{\scriptsize to }i}$ of $M_j$. Although many vertices in $M_j$ may send request to $u$, only one request to $u$ will be sent from $M_j$ to $M_i$ since $S_{\mbox{\scriptsize to }i}$ is a (hash) set that eliminates redundant elements.

\begin{figure}[!t]
    \centering
    \includegraphics[width=\columnwidth]{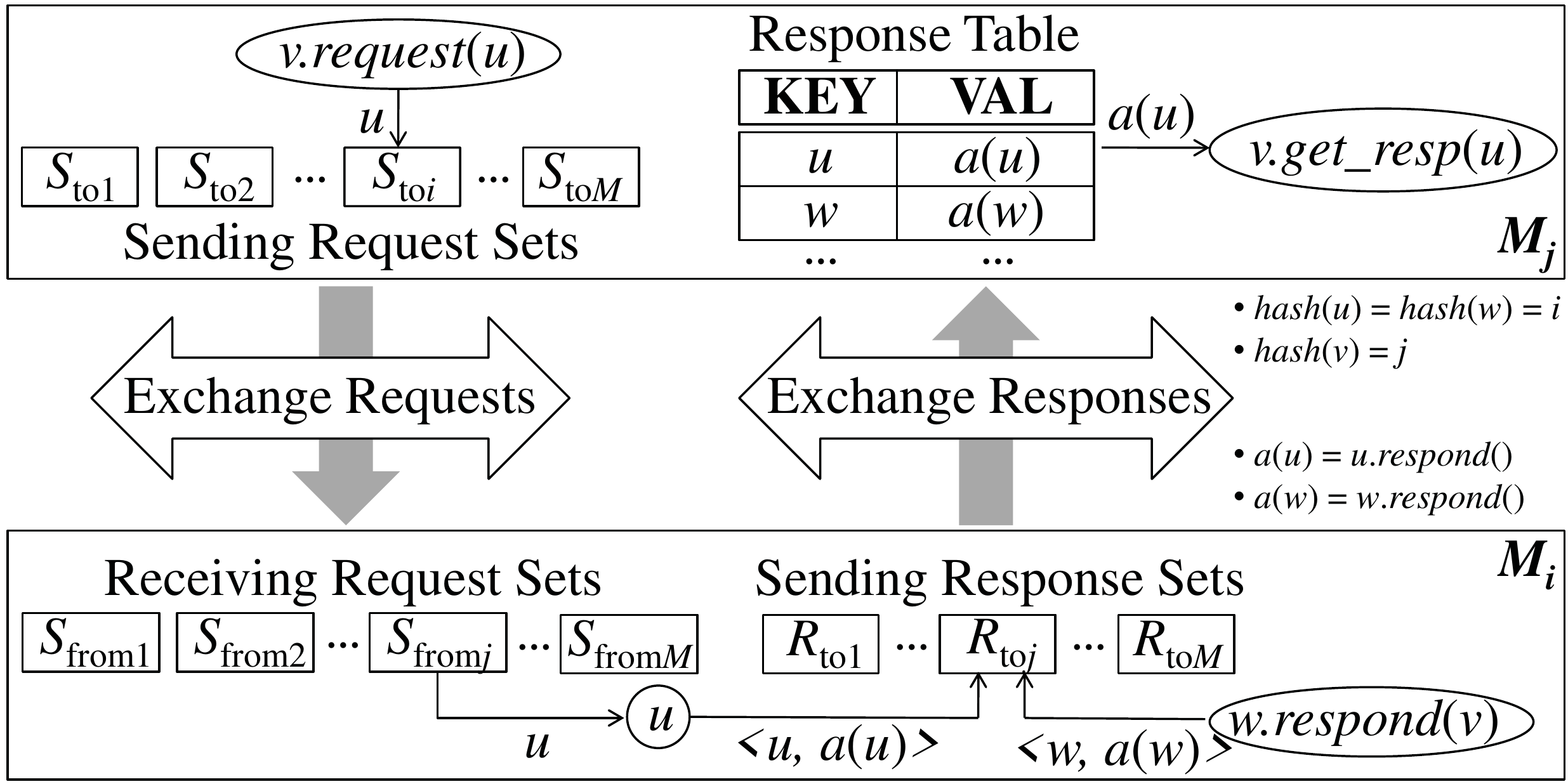}
    \vspace{-4mm}
    \caption{Illustration of request-respond paradigm}\label{reqresp}
    \vspace{-2mm}
\end{figure}

After {\em compute}() is called for all active vertices, the vertex-to-vertex messages are first exchanged through $Ch_{msg}$. Then, each machine sends each request set $S_{\mbox{\scriptsize to }k}$ to machine $M_k$. After the requests are exchanged, each machine receives $M$ request sets, where set $S_{\mbox{\scriptsize from }k}$ stores the requests sent from machine $M_k$. In the example shown in Figure~\ref{reqresp}, $u$ is contained in the set $S_{\mbox{\scriptsize from} j}$ in machine $M_i$, since vertex $v$ in machine $M_j$ sent request to $u$.

Then, a response set $R_{\mbox{\scriptsize to} k}$ is constructed for each request set $S_{\mbox{\scriptsize from }k}$ received, which is to be sent back to machine $M_k$. In our example, the requested vertex, $u \in S_{\mbox{\scriptsize from} j}$, calls a user-specified function {\em respond}() to return its specified attribute $a(u)$, and adds the entry $\langle u, a(u)\rangle$ to the response set $R_{\mbox{\scriptsize to} j}$.

Once the response sets are exchanged, each machine constructs a hash table from the received entries. In the example shown in Figure~\ref{reqresp}, the entry $\langle u, a(u)\rangle$ is received by machine $M_j$ since it is in the response set $R_{\mbox{\scriptsize to} j}$ in machine $M_i$. The hash table is available for the next superstep, where vertices can access their requested value in their {\em compute}() function. In our example, vertex $v$ in machine $M_j$ may call {\em get\_resp}($u$) in the next superstep, which looks up $u$'s attribute $a(u)$ from the hash table.

The following theorem shows the effectiveness of the request-respond paradigm for message reduction.


\begin{theorem}\label{th:reqBound}
Let $\{v_1, v_2, \ldots, v_\ell\}$ be the set of requesters that request the attribute $a(u)$ from a vertex $u$. Then, the request-respond paradigm reduces the total number of messages from $2\ell$ in Pregel's vertex-to-vertex message passing framework to $2\min(M,\ell)$, where $M$ is the number of machines.
\end{theorem}
\begin{proof}
The proof follows directly from the fact that each machine sends at most 1 request to $u$ even though there may be more than 1 requester in that machine, and that at most 1 respond from $u$ is sent to each machine that makes a request to $u$, and that there are at most $\min(M,\ell)$ machines that contain a requester.
\end{proof}


In the worst case, the request-respond paradigm uses the same number of messages as Pregel's vertex-to-vertex message passing. But in practice, many Pregel algorithms (e.g., those described in Sections~\ref{ssec:sv} and~\ref{ssec:msf}) have bottleneck vertices with a large number of requesters, leading to imbalanced workload and long elapsed running time. In such cases, our request-respond paradigm effectively bounds the number of messages to the number of machines containing the requesters and eliminates the imbalanced workload.

\vspace{2mm}

\noindent{\bf Explicit Responding.} In the above discussion, a vertex $v$ simply calls {\em request}($u$) in one superstep, and it can then call {\em get\_resp}($u$) in the next superstep to get $a(u)$. All the operations including request exchange, response set construction, response exchange, and response table construction are performed by Pregel+ automatically and are thus transparent to users. We name the above process as \emph{implicit responding}, where a responder does not know the requester until a request is received.

When a responder $w$ knows its requesters $v$, $w$ can explicitly call {\em respond}($v$) in {\em compute}(), which adds $\langle w, w.${\em respond}()$\rangle$ to the response set $R_{\mbox{\scriptsize to }j}$ where $j=hash(v)$. This process is also illustrated in Figure~\ref{reqresp}. Explicit responding is more cost-efficient since there is no need for request exchange and response set construction.

Explicit responding is useful in many applications. For example, to compute PageRank on an \emph{undirected} graph, a vertex $v$ can simply call {\em respond}($u$) for each $u\in\Gamma(v)$ to push $a(v)=pr(v)/|\Gamma(v)|$ to $v$'s neighbors; this is because in the next superstep, vertex $u$ knows its neighbors $\Gamma(u)$, and can thus collect their responses. Similarly, in {\em attribute broadcast}, if the input graph is undirected, each vertex $v$ can simply push its attribute $a(v)$ to its neighbors. Note that data pushing by explicit responding requires less messages than by Pregel's vertex-to-vertex message passing, since responds are sent to machines (more precisely, their response tables) rather than individual vertices. 

\vspace{2mm}

\noindent{\bf Programming Interface.} Pregel+ extends the vertex class in Pregel's interface~\cite{pregel} by requiring users to specify an additional template argument $<${\em R}$>$, which indicates the type of the attribute value that a vertex responds.

In {\em compute}(), a vertex can either pull data from another vertex $v$ by calling {\em request}($v$), or push data to $v$ by calling {\em respond}($v$). The attribute value that a vertex returns is defined by a user-specified abstract function {\em respond}(), which returns a value of type $<${\em R}$>$. Like {\em compute}(), one may program {\em respond}() to return different attributes of a vertex in different supersteps according to the algorithm logic of the specific application. Finally, a vertex may call {\em get\_resp}($v$) in {\em compute}() to get the attribute of $v$, if it is pushed into the response table in the previous superstep.

\section{Experimental Results}\label{sec:results}
We now evaluate the effectiveness of our message reduction techniques. We ran our experiments on a cluster of 16 machines, each with 24 processors (two Intel Xeon E5-2620 CPU) and 48GB RAM. One machine is used as the master, while the other 15 machines act as slaves. The connectivity between any pair of nodes in the cluster is 1Gbps.

We used five real-world datasets, as shown in Figure~\ref{data}: (1){\em WebUK}\footnote[2]{http://law.di.unimi.it/webdata/uk-union-2006-06-2007-05}: a web graph generated by combining twelve monthly snapshots of the .uk domain collected for the DELIS project; (2){\em LiveJournal} ({\em LJ}) \footnote[3]{http://konect.uni-koblenz.de/networks/livejournal-groupmemberships}: a bipartite network of LiveJournal users and their group memberships; (3){\em Twitter}\footnote[4]{http://konect.uni-koblenz.de/networks/twitter\_mpi}: Twitter who-follows-who network based on a snapshot taken in 2009; (4){\em BTC}\footnote[5]{http://km.aifb.kit.edu/projects/btc-2009/}: a semantic graph converted from the Billion Triple Challenge 2009 RDF dataset; (5)\emph{USA}\footnote[6]{http://www.dis.uniroma1.it/challenge9/download.shtml}: the USA road network.

\begin{figure}[!t]
    \centering
    \includegraphics[width=\columnwidth]{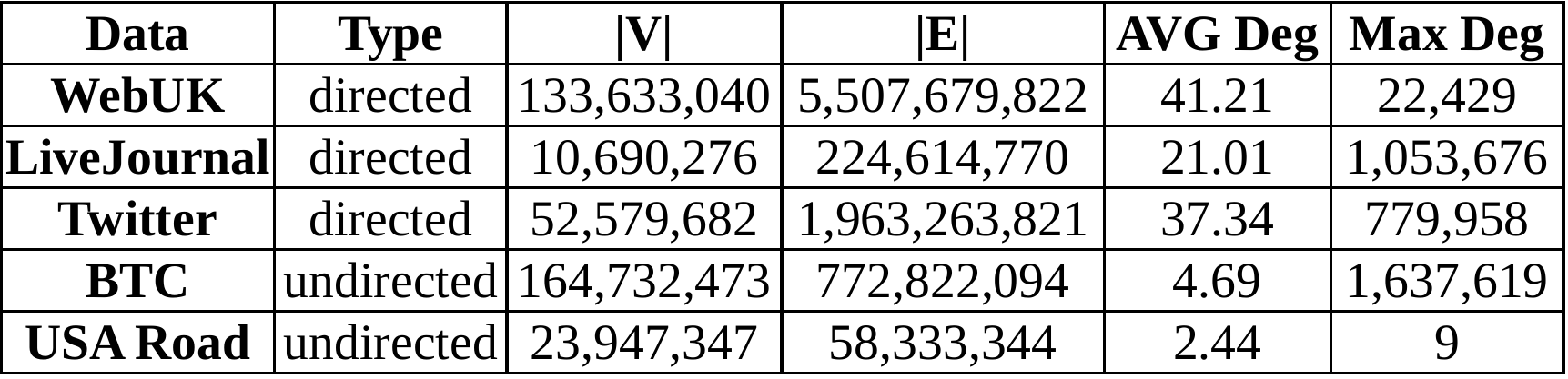}
    \vspace{-5mm}
    \caption{Datasets (M = million)}\label{data}
    \vspace{-4mm}
\end{figure}

\begin{figure*}[!t]
 \centering
\includegraphics[width=1.9\columnwidth]{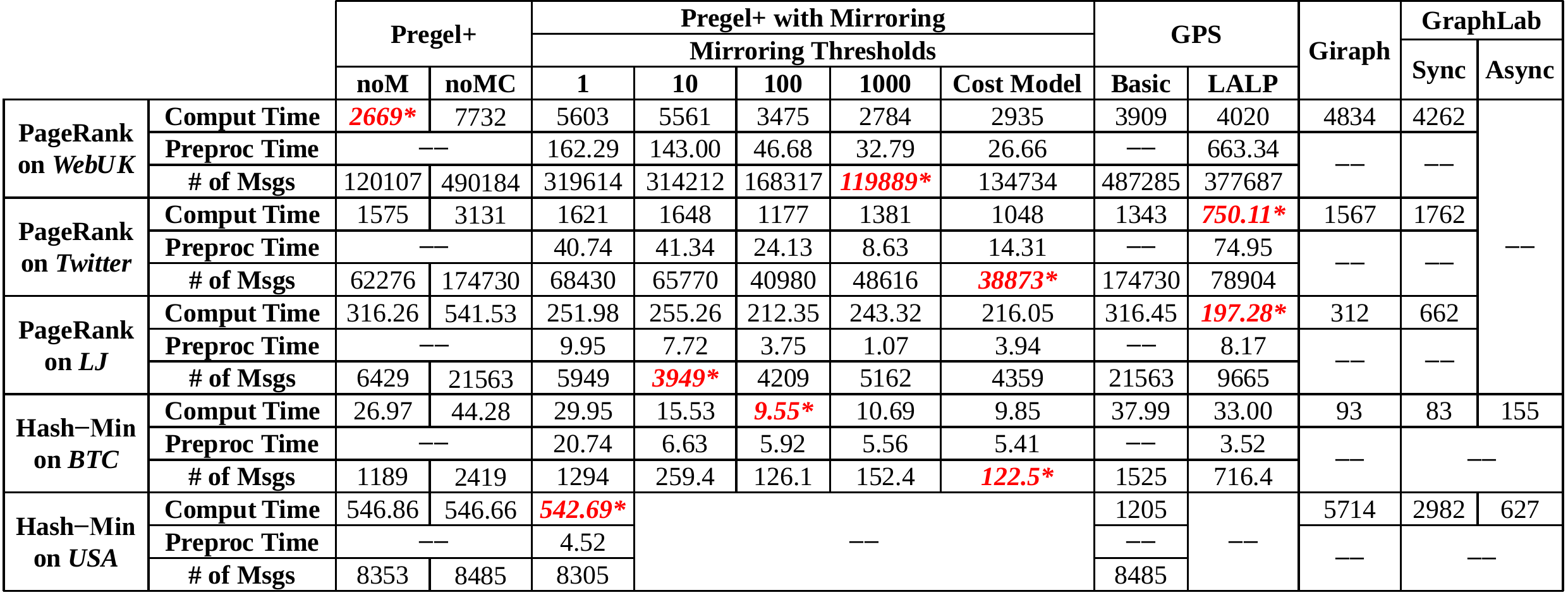}
\caption{Effects of mirroring ($\star$: best result; Comput/Preproc time: Computation/Preprocessing time in sec; \# of Msgs: \# of messages in millions)}\label{mirexp}
\vspace{-3mm}
\end{figure*}

{\em LJ}, {\em Twitter} and {\em BTC} have skewed degree distribution; {\em WebUK}, {\em LJ} and {\em Twitter} have relatively high average degree; \emph{USA} and \emph{WebUK} have a large diameter.

\vspace{2mm}

\noindent \textbf{Pregel+ Implementation. } Pregel+ is implemented in C/C++ as a group of header files, and users only need to include the necessary base classes and implement the application logic in their subclasses. Pregel+ communicates with HDFS through libhdfs, a JNI based C API for HDFS. Each worker is simply an MPI process and communications are implemented using MPI communication primitives. While one may deploy Pregel+ with any Hadoop and MPI version, we use Hadoop 1.2.1 and MPICH 3.0.4 in our experiments. All programs are compiled using GCC 4.4.7 with -O2 option enabled.

All the system source codes, as well as the source codes of the algorithms discussed in this paper, can be found in {\small \url{http://www.cse.cuhk.edu.hk/pregelplus}}.

\subsection{Effectiveness of Mirroring}    \label{result:mirror}
Figure~\ref{mirexp} reports the performance gain by mirroring. We measure the gain by comparing with (1)Pregel+ without both mirroring and combiner, denoted by \textbf{Pregel-noMC}; (2)Pregel+ with combiner but without mirroring, denoted by \textbf{Pregel-noM}; and (3)GPS~\cite{gps} with and without LALP.  The request-respond technique is not applied in Pregel+ for this set of experiments. As a reference, we also report the performance of Giraph~1.0.0 \cite{giraph} (with combiner) and GraphLab~2.2 (which includes PowerGraph~\cite{powergraph}).

We test the mirroring thresholds 1, 10, 100, 1000, and the one automatically set by the cost model given by Theorem~\ref{th:mirror} (which is 199, 165, 62, 126, for {\em WebUK}, {\em Twitter}, {\em LJ}, {\em BTC}, respectively). But for the {\em USA} road network, its maximum vertex degree is only 9 and thus we do not apply mirroring with large thresholds. For GPS, we follow~\cite{sys_vldb} and fix the threshold of LALP as 100. This is a reasonable choice, since \cite{sys_vldb} reports that this threshold achieves good performance in general, and we find that the best performance after tuning the threshold is very close to the performance when the threshold is 100. We also report the preprocessing time of constructing mirrors for Pregel+ and that of LALP for GPS in rows marked by ``Preproc Time''. We also report the number of messages sent by Pregel+ and GPS (note that Giraph does not report the number of messages, but the number should be the same as that of Pregel-noMC and Pregel-noM; while GraphLab does not employ message passing).

We ran PageRank on the three directed graphs, and {\em Hash-Min} on the two undirected graphs in Figure~\ref{data}. For PageRank computation, we use aggregator to check whether every vertex changes its PageRank value by less than 0.01 after each superstep, and terminate if so. The computation takes 89, 89 and 96 supersteps on {\em WebUK}, {\em Twitter} and {\em LJ}, respectively, before convergence. We do not run GraphLab in asynchronous mode for PageRank, since its convergence condition is different from the synchronous version and hence leads to different PageRank results.

\vspace{2mm}

\noindent \textbf{Mirroring in Pregel+. } As Figure~\ref{mirexp} shows, mirroring significantly improves the performance of Pregel-noM, in terms of the reduction in both running time and message number. The improvement is particularly obvious for the graphs, {\em Twitter}, {\em LJ}, and {\em BTC}, which have highly skewed degree distribution. Thus, the result also demonstrates the effectiveness of mirroring in workload balancing.

Mirroring is not so effective for PageRank on {\em WebUK}, for which Pregel-noM has the best performance. The number of messages is only slightly decreased when mirroring threshold $\tau= 1000$, and yet it is still slower than Pregel-noM. This is because messages sent through $Ch_{mir}$ are intercepted by mirrors which incurs additional cost. Since the degree of the majority of the vertices in {\em WebUK} is not very high, mirroring does not significantly reduce the number of messages, and thus, the additional cost of $Ch_{mir}$ is not paid off.

The results also show that the mirroring threshold given by our cost model achieves either the best performance, or close to the performance of the best threshold tested. The one-off preprocessing time required to construct the mirrors is also short compared with the computation time.

\vspace{2mm}

\noindent \textbf{Comparison with Other Systems. } Figure~\ref{mirexp} shows that Pregel+ without mirroring (i.e., Pregel-noM) is already faster than both Giraph and GraphLab, which verifies that our $Ch_{msg}$ implementation is efficient, and thus the performance gain by mirroring is not an over-claimed improvement gained over a slow implementation.

Compared with GPS, the reduction in both message number and running time achieved by \emph{the integration of mirroring and combiner} in Pregel+ is significantly more than that achieved by \emph{LALP alone} in GPS, which can be observed from (1)Pregel+ with mirroring vs.\ Pregel-noMC, and (2)GPS with LALP v.s.\ GPS without LALP. In contrast to the claim in~\cite{gps} that message combining is not effective, our result clearly demonstrates the benefits of integrating mirroring and combiner, and hence highlights the importance of our theoretical analysis on the tradeoff between mirroring and  message combining (i.e., Theorem~\ref{th:mirror}).

However, we notice that GPS is sometimes faster than Pregel+ even though much more messages are exchanged. We found it hard to explain and so we studied the codes of GPS to explore the reason, which we explain below. GPS requires that vertex IDs should be integers that are contiguous starting from $0, 1, \cdots, |V|$; while other systems allow vertex IDs to be of any user-specified type as long as a hash function is provided (for calculating the ID of the worker that a vertex resides in). As a result of the dense ID representation, each worker in GPS simply maintains the incoming message buffers of the vertices by an array, and when a worker receives a message targeted at vertex $tgt$, it is put into $tgt$'s incoming message buffer (i.e., $I_{tgt}$) whose position in the array can be directly computed from $tgt$. On the contrary, systems like Pregel+ and Giraph need to look up $I_{tgt}$ from a hash table using key $tgt$, which has extra cost for each message exchanged.

We remark that there are good reasons to require vertex IDs to take arbitrary type, rather than to hard-code them as contiguous integers. For example, the Pregel algorithm in~\cite{ppa_vldb} for computing bi-connected components constructs an auxiliary graph from the input graph, and each vertex of the auxiliary graph corresponds to an edge $(u, v)$ of the input graph. While we can simply use integer pair as vertex ID in Pregel+, using GPS requires extra effort from programmers to relabel the vertices of the auxiliary graph with contiguous integer IDs, which can be costly for a large graph. We note that, if one desires, he can easily implement GPS's dense vertex ID representation in Pregel+ to further improve the performance for certain algorithms, but this is not the focus of our work which studies message reduction techniques.

\subsection{Effectiveness of Request-Respond Technique}\label{result:req}
\begin{figure}[!t]
    \centering
    \includegraphics[width=\columnwidth]{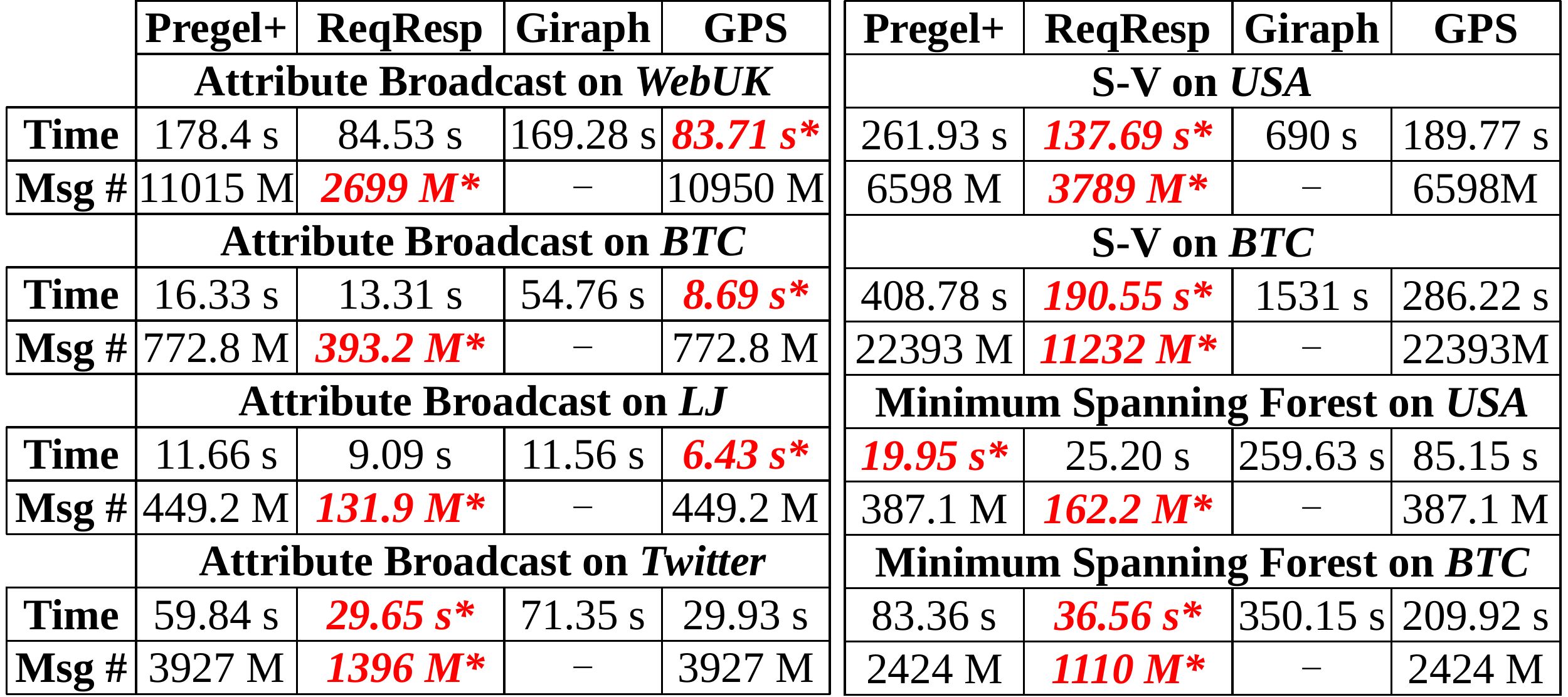}
        \vspace{-3mm}
    \caption{Effects of the request-respond technique}\label{reqexp}
    \vspace{-3mm}
\end{figure}

Figure~\ref{reqexp} reports the performance gained by the request-respond technique. We test the three algorithms in Section~\ref{sec:alg} to which the request-respond technique is applicable: {\em attribute broadcast}, {\em S-V} and minimum spanning forest. We also include Giraph and GPS as a reference. We do not include GraphLab since the algorithms cannot be easily implemented in GraphLab (e.g., it is not clear how a vertex $v$ can communicate with a non-neighbor $D[v]$ as in {\em S-V} and minimum spanning forest).

The results show that  Pregel+ with request-respond, denoted by \textbf{ReqResq}, uses significantly less messages. For example, for {\em attribute broadcast} on {\em WebUK}, ReqResq reduces the message number from 11,015 million to only 2,699 million. ReqResq also records the shortest running time except in a few cases where GPS is faster due to the same reason given in Section~\ref{result:mirror}. Another exception is when computing minimum spanning forest on {\em USA}, where Pregel+ is faster without request-respond. This is because vertices in {\em USA} have very low degree, rendering the request-respond technique ineffective, and the additional computational overhead is not paid off by the reduction in message number.

\section{Conclusions}  \label{sec:conclude}

We presented two techniques to reduce the amount of communication and to eliminate skewed communication workload. The first technique, mirroring, eliminates communication bottlenecks caused by high vertex degree, and is transparent to programming. The second technique is a new request-respond paradigm, which eliminates bottlenecks caused by program logic, and simplifies the programming of many Pregel algorithms. Our experiments on large real-world graphs verified that our techniques are effective  in reducing the communication cost and overall computation time.

{\small

\bibliographystyle{abbrv}

\bibliography{ref_pullgel}
}

\end{document}